\input ./style/arxiv-general.cfg
\documentclass[aos,MSNbibl,dvips]{arximspdf}
\makeatletter
   \@ifpackageloaded{graphicx}{}{\usepackage{graphicx}}
\makeatother
\usepackage{mathbh,dcolumn}

\doi{10.1214/15-AOS1350}
\volume{43}
\issue{6}
\pubyear{2015}
\firstpage{2507}
\lastpage{2536}
\docsubty{FLA}

\makeatletter
\newcolumntype{d}[1]{D{.}{.}{#1}}
\def\sfrac#1#2{#1/#2}

\def\sklfrac#1#2{(#1/#2)}

\def\sklafrac#1#2{(#1/(#2))}

\newcommand{\rrvert}{\vert}
\newcommand{\rrVert}{\Vert}
\newcommand{\llvert}{\vert}
\newcommand{\llVert}{\Vert}
\newtheorem{theorem}{Theorem}
\newtheorem{lemma}[theorem]{Lemma}
\newtheorem{proposition}[theorem]{Proposition}
\newtheorem{prop}[theorem]{Proposition}
\newtheorem{teo}[theorem]{Theorem}
\newtheorem{cor}[theorem]{Corollary}
\newtheorem{lem}[theorem]{Lemma}
\newproclaim{rem}{Remark}

\renewcommand{\mid}{|}

\newcommand{\field}[1]{\mathbb{#1}}
\newcommand{\R}{\field{R}}
\newcommand{\E}{\field{E}}
\newcommand{\EXP}{\E}
\newcommand{\PROB}{\field{P}}
\newcommand{\argmin}{\mathop{\operatorname{arg}\operatorname{min}}}
\newcommand{\esssup}{\mathop{\operatorname{ess}\operatorname{sup}}}
\newcommand{\C}{\mathcal{C}}
\newcommand{\X}{\mathcal{X}}
\newcommand{\Z}{\mathcal{Z}}
\newcommand{\ol}{\overline}
\newcommand{\wh}{\widehat}
\newcommand{\wb}{\overline}
\newcommand{\hf}{\wh{f}}
\newcommand{\hr}{\wh{r}}
\newcommand{\br}{\wb{r}}
\newcommand{\ep}{\varepsilon}
\newcommand{\F}{\mathcal{F}} 
\newcommand{\G}{\mathcal{G}} 
\newcommand{\N}{\mathbb{N}} 
\makeatother

\begin{document}
\begin{frontmatter}

\title{Empirical risk minimization for heavy-tailed losses}
\runtitle{Risk minimization for heavy tails}

\begin{aug}
\author[A]{\fnms{Christian}~\snm{Brownlees}\thanksref{T1}\ead[label=e1]{christian.brownlees@upf.edu}},
\author[B]{\fnms{Emilien}~\snm{Joly}\ead[label=e2]{emilien.joly@ens.fr}}
\and
\author[C]{\fnms{G\'abor}~\snm{Lugosi}\corref{}\thanksref{T1}\ead[label=e3]{gabor.lugosi@upf.edu}}
\runauthor{C. Brownlees, E. Joly and G. Lugosi}
\affiliation{Pompeu Fabra University, HEC Paris--CNRS and Pompeu Fabra University}
\address[A]{C. Brownlees\\
Department of Economics and Business\\
Pompeu Fabra University\\
Ramon Trias Fargas 25-27\\
08005 Barcelona\\
Spain\\
\printead{e1}}
\address[B]{E.~Joly\\
GREGHEC\\
HEC Paris--CNRS\\
1 rue de la lib\'eration 78350\\
Jouy-en-Josas\\
France\\
\printead{e2}}
\address[C]{G. Lugosi\\
ICREA and Department of Economics and Business\\
Pompeu Fabra University\\
Ramon Trias Fargas 25-27\\
08005 Barcelona\\
Spain\\
\printead{e3}}
\end{aug}
\thankstext{T1}{Supported by the Spanish Ministry of Science and
Technology Grant MTM2012-37195.}

%
\received{\smonth{6} \syear{2014}}
%
\revised{\smonth{5} \syear{2015}}

%
\begin{abstract}
The purpose of this paper is to discuss empirical risk
minimization when the losses are not necessarily bounded
and may have a distribution with heavy tails. In such
situations, usual empirical averages may fail to provide reliable
estimates and empirical risk minimization may provide large excess risk.
However, some robust mean estimators proposed in the literature
may be used to replace empirical means.
In this paper, we investigate empirical risk minimization
based on a robust estimate proposed by Catoni.
We develop performance bounds based on chaining arguments
tailored to Catoni's mean estimator.
\end{abstract}

%
\end{frontmatter}

\section{Introduction}

Heavy-tailed data are commonly encountered in many fields of research
(see, e.g., Embrechts, Kl\"uppelberg and Mikosch \cite{Emb97} and
Finkenstadt and Rootz\'en \cite{Fin03}).
For instance, in finance, the influential work of Mandelbrot \cite
{Man63} and Fama \cite{Fam63} documented evidence of power-law
behavior in asset prices in the early 1960s.
When the data have heavy tails, standard statistical procedures
typically perform poorly and
appropriate robust alternatives are needed to carry out inference effectively.
In this paper, we propose a class of robust empirical risk minimization
procedures for such data that are based on a robust estimator
introduced by Catoni \cite{Cat10}.

Empirical risk minimization is one of the basic principles of
statistical learning that is
routinely applied in a great variety of problems such as regression
function estimation, classification and clustering.
The general model may be described as follows.
Let $X$ be a random variable taking values in some measurable space
$\mathcal{X}$ and let $\F$ be a set of nonnegative functions
defined on $\X$.
For each $f \in\F$, define the \textit{risk} $m_f=\EXP f(X)$ and let
$m^*=\inf_{f\in\F} m_f$ denote the optimal risk.
In statistical learning, $n$ independent random variables
$X_1,\ldots,X_n$ are available, all distributed as $X$, and
one aims at finding a function with small risk. To this end,
one may define the \textit{empirical risk minimizer}
\[
f_{\mathrm{ERM}}= \argmin_{f\in\F} \frac{1}{n} \sum
_{i=1}^n f(X_i) ,
\]
where, for the simplicity of the discussion and essentially
without loss of generality,
we implicitly assume that the minimizer exists.
If the minimum is achieved by more than one function, one may
pick one of them arbitrarily.

\begin{rem*}[(Loss functions and risks)]
The main motivation and terminology may be explained by
the following general prediction problem in statistical learning.
Let the ``training data'' $(Z_1,Y_1),\ldots,(Z_n,Y_n)$ be independent
identically
distributed pairs of random variables
where the $Z_i$ take their
values in, say, $\R^m$ and the $Y_i$ are real-valued.
In classification problems, the $Y_i$ take discrete values.
Given a new observation $Z$, one is interested in predicting the
value of the corresponding response variable $Y$ where the pair
$(Z,Y)$ has the same distribution as that of the $(Z_i,Y_i)$.
A predictor is a function $g:\R^m \to\R$ whose quality is measured
with the help of a \textit{loss function} $\ell: \R\times\R\to\R_+$.
The \textit{risk} of $g$ is then $\EXP\ell(g(Z),Y)$.
Given a class $\G$ of functions $g:\R^m \to\R$, empirical risk
minimization chooses one that minimizes the \textit{empirical risk}
$(1/n) \sum_{i=1}^n \ell(g(Z_i),Y_i)$ over all $g\in\G$.
In the simplified notation followed in this paper, $X_i$ corresponds to
the pair $(Z_i,Y_i)$, the function $f$ represents $\ell(g(\cdot
),\cdot)$
and $m_f$ substitutes $\EXP\ell(g(Z),Y)$.
\end{rem*}

The performance of empirical risk minimization is measured by the
\textit{risk} of the selected function,
\[
m_{\mathrm{ERM}}= \EXP\bigl[ f_{\mathrm{ERM}}(X) \mid X_1,
\ldots,X_n \bigr].
\]
In particular, the main object of interest for this paper is
the \textit{excess risk} $m_{\mathrm{ERM}}-m^*$.
The performance of empirical risk minimization has been
thoroughly studied and well understood using tools of empirical process
theory. In particular, the simple observation that
\[
m_{\mathrm{ERM}}-m^* \le2\sup_{f\in\F} \Biggl\llvert
\frac{1}{n}\sum_{i=1}^n
f(X_i)-m_f\Biggr\rrvert,
\]
allows one to apply the rich theory on the suprema of empirical
processes to obtain upper performance bounds.
The interested reader is referred to
Bartlett and Mendelson \cite{BaMe06},
Boucheron, Bousquet and Lugosi \cite{BoBoLu05},
Koltchinskii \cite{Kol06},
Massart \cite{Mas06},
Mendelson \cite{Men14},
van de Geer \cite{Geer00}
for references and recent results in this area.
Essentially all of the theory of empirical minimization
assumes either that the functions $f$ are uniformly bounded
or that the random variables $f(X)$ have sub-Gaussian tails for
all $f\in\F$. For example, when all $f\in\F$ take their values
in the interval $[0,1]$, Dudley's \cite{Dud78}
classical metric-entropy bound, together with standard symmetrization
arguments, imply that there exists a
universal constant $c$ such that
%
\begin{equation}
\label{eqdudley1} \EXP m_{\mathrm{ERM}}-m^* \le\frac{c}{\sqrt{n}} \EXP
\int
_0^1 \sqrt{\log N_{\mathbb{X}}(\F,
\epsilon)}\,d\epsilon,
\end{equation}
where for any $\epsilon>0$,
$N_{\mathbb{X}}(\F,\epsilon)$ is the {\sl$\epsilon$-covering number}
of the class $\F$ under the empirical quadratic distance
$d_{\mathbb{X}}(f,g)= (\frac{1}{n}\sum_{i=1}^n
(f(X_i)-g(X_i))^2 )^{1/2}$,
defined as the minimal cardinality
$N$ of any set $\{f_1,\ldots,f_N\}\subset\F$ such that for all
$f\in\F$ there exists an $f_j\in\{f_1,\ldots,f_N\}$ with
$d_{\mathbb{X}}(f,f_j)\le\epsilon$.
Of course, this is one of the most basic bounds and many important
refinements have been established.

A tighter bound may be established by the so-called \textit{generic chaining}
method; see Talagrand \cite{talagrand}.
Recall the following definition (see, e.g., \cite{talagrand}, Definition~1.2.3).
Let $T$ be a (pseudo) metric space. An increasing sequence
$(\mathcal{A}_n)$ of partitions of $T$ is called \textit{admissible} if
for all $n=0,1,2,\ldots, \#\mathcal{A}_n \le2^{2^n}$.
For any $t\in T$, denote by $A_n(t)$ the unique element
of $\mathcal{A}_n$ that contains $t$.
Let $\Delta(A)$ denote the diameter of the set $A\subset T$.
Define, for $\beta=1,2$,
\[
\gamma_{\beta}(T,d) =\inf_{\mathcal{A}_n}\sup
_{t \in T} \sum_{n \ge0} 2^{n/\beta
}
\Delta\bigl(A_n(t)\bigr),
\]
%
where the infimum is taken over all admissible sequences.
Then one has
%
\begin{equation}
\label{eqclassical} \EXP m_{\mathrm{ERM}}-m^* \le\frac{c}{\sqrt{n}}\EXP
\gamma_2(\F,d_{\mathbb{X}}) ,
\end{equation}
for some universal constant $c$. This bound implies
(\ref{eqdudley1}) as
$\gamma_2(\F,d_{\mathbb{X}})$ is bounded by a constant multiple of
the entropy integral
$\int_0^1 \sqrt{\log N_{\mathbb{X}}(\F,\epsilon)}\,d\epsilon$ (see,
e.g., \cite{talagrand}).

However, when the functions $f$ are no longer uniformly bounded and
the random variables $f(X)$ may have a heavy tail,
empirical risk minimization may have a much poorer performance.
This is simply due to the fact that empirical averages become
poor estimates of expected values. Indeed, for heavy-tailed distributions,
several estimators of the mean are known to outperform simple empirical
averages. It is a natural idea to define
a robust version of empirical risk minimization based on minimizing
such robust estimators.

In this paper, we focus on an elegant and powerful estimator
proposed and analyzed by Catoni \cite{Cat10}. (A version of)
Catoni's estimator may be defined as follows.

Introduce the nondecreasing differentiable \textit{truncation function}
%
\begin{equation}
\label{eqtrunc} \phi(x) = -\mathbh{1}_{\{ x<0 \}}\log\biggl(1-x+
\frac{x^2}{2} \biggr) +\mathbh{1}_{\{ x\ge0 \}}\log\biggl(1+x+
\frac{x^2}{2} \biggr).
\end{equation}
To estimate $m_f=\EXP f(X)$ for some $f\in\F$, define for
all $\mu\in\R$,
\[
\wh{r}_f(\mu) = \frac{1}{n\alpha} \sum
_{i=1}^n\phi\bigl(\alpha\bigl(f(X_i)-
\mu\bigr)\bigr),
\]
where $\alpha>0$ is a parameter of the estimator to be specified below.
Catoni's estimator of $m_f$ is defined as the
unique value $\wh{\mu}_f$ for which
$\wh{r}_f (\wh{\mu}_f)=0$.
[Uniqueness is ensured by the strict monotonicity of $ \mu\mapsto\wh
{r}_f(\mu) $.]
Catoni proves that for any fixed $f\in\F$ and $\delta\in[0,1]$
such that $n>2\log(1/\delta)$, under the only
assumption that $\operatorname{Var} (f(X) ) \le v$, the
estimator above with
\[
\alpha= \sqrt{\frac{2\log(1/\delta)}{
n (v+\sklafrac{2v\log(1/\delta)}{n(1-(2/n)\log(1/\delta))} )}}
\]
satisfies that, with probability at least $1-2\delta$,
%
\begin{equation}
\label{eqcatoni0} \llvert m_f-\wh{\mu}_f\rrvert\le\sqrt{
\frac{2v\log
(1/\delta
)}{n(1-(2/n)\log(1/\delta))}}.
\end{equation}
In other words, the deviations of the estimate exhibit a sub-Gaussian
behavior. The price to pay is that the estimator depends both on
the upper bound $v$ for the variance and on
the prescribed confidence $\delta$ via the parameter $\alpha$.

Catoni also shows that for any $n>4(1+\log(1/\delta))$, if
$\operatorname{Var} (f(X) ) \le v$, the choice
\[
\alpha= \sqrt{\frac{2}{nv}}
\]
guarantees that, with probability at least $1-2\delta$,
%
\begin{equation}
\label{eqcatoni} \llvert m_f-\wh{\mu}_f\rrvert\le
\bigl(1+\log(1/\delta) \bigr)\sqrt{\frac{v}{n}}.
\end{equation}
Even though we lose the sub-Gaussian tail behavior, the estimator
is independent of the required confidence level.

Given such a powerful mean estimator, it is natural to propose
an empirical risk minimizer that
selects a function from the class $\F$ that minimizes Catoni's
mean estimator.
Formally, define
\[
\wh{f}=\argmin_{f \in\F} \wh{\mu}_f,
\]
where again, for the sake of simplicity we assume that the minimizer
exists. (Otherwise one may select an appropriate approximate minimizer
and all arguments go through in a trivial way.)

Once again, as a first step of understanding the excess risk
$m_{\wh{f}}-m^*$, we may use the simple bound
\[
m_{\wh{f}}-m^* = (m_{\wh{f}}- \wh{\mu}_{\wh{f}} ) + \bigl(
\wh{\mu}_{\wh{f}}- m^* \bigr) \le2\sup_{f\in\F} \llvert
m_f-\wh{\mu}_f\rrvert.
\]
When $\F$ is a finite class of cardinality, say $\llvert \F\rrvert =N$,
Catoni's bound may be combined, in a straightforward way,
with the union-of-events bound. Indeed, if the estimators $\wh{\mu}_f$
are defined with parameter
\[
\alpha= \sqrt{\frac{2\log(N/\delta)}{
n (v+\sklafrac{2v\log(N/\delta)}{n(1-(2/n)\log(N/\delta))} )}},
\]
then, with probability at least $1-2\delta$,
\[
\sup_{f\in\F} \llvert m_f-\wh{\mu}_f
\rrvert\le\sqrt{\frac{2v\log(N/\delta)}{n(1-(2/n)\log(N/\delta))}}.
\]
Note that this bound requires that $\sup_{f\in\F} \operatorname
{Var} (f(X) ) \le v$,
that is, the variances are uniformly bounded by a \textit{known} value $v$.
Throughout the paper, we work with this assumption.
However, this bound does not take into account the structure of the
class $\F$ and it is useless when $\F$ is an infinite class.
Our strategy to obtain meaningful bounds is to use
\textit{chaining} arguments. However, the extension is nontrivial and
the argument becomes more involved. The main results of the
paper present performance bounds for empirical minimization
of Catoni's estimator based on generic chaining.

\begin{rem*}[(Median-of-means estimator)]
Catoni's estimator is not the only one with sub-Gaussian deviations
for heavy-tailed distributions. Indeed, the \textit{median-of-means} estimator,
proposed by Nemirovsky and Yudin \cite{NeYu83} (and also independently by
Alon, Matias and Szegedy \cite{AMS02}) has similar performance guarantees
as (\ref{eqcatoni0}). This estimate is obtained by dividing the
data in several small blocks, calculating the sample mean within
each block, and then taking the median of these means.
Hsu and Sabato \cite{HsuSa13} and Minsker \cite{Min13}
introduce multivariate generalizations of the median-of-means
estimator and use it to define and analyze certain statistical
learning procedures in the presence of heavy-tailed data.
The sub-Gaussian behavior is
achieved under various assumptions on the
loss function.
Such conditions can be avoided here.
As an example, we detail applications of our results in Section~\ref
{secapps} for three different examples of loss functions.
An important advantage of the median-of-means estimate over Catoni's
estimate is that
the parameter of the estimate (i.e., the number of blocks) only depends
on the confidence level $\delta$ but not on $v$ and, therefore, no prior
upper bound of the variance $v$ is required to compute this estimate.
Also, the median-of-means estimate is useful even when the variance
is infinite and only a moment of order $1+\epsilon$ exists for some
$\epsilon>0$ (see Bubeck, Cesa-Bianchi and Lugosi \cite{BuCeLu13}).
Lerasle and Oliveira \cite{LeOl12} consider
empirical minimization of the median-of-means
estimator and obtain interesting results in various statistical
learning problems.
However, to establish
metric-entropy bounds for minimization of this mean estimate
remains to be a challenge.
\end{rem*}


The rest of the paper is organized as follows. In Section~\ref{secresults},
we state and discuss the main results of the paper. Section~\ref{secproofs}
is dedicated to the proofs.
In Section~\ref{secapps}, we describe some applications
to regression under the absolute and squared losses and $k$-means
clustering.
Finally, in Section~\ref{secsim} we present some simulation results
both for
regression and $k$-means clustering.
The simulation study gives empirical evidence that the proposed
empirical risk minimization procedure improves performance in a
significant manner in the presence of heavy-tailed data.
Some of the more technical arguments are relegated to the \hyperref[app]{Appendix}.

\section{Main results}
\label{secresults}

The bounds we establish for the excess risk
depend on the geometric structure of the
class $\F$ under different distances. The $L_2(P)$ distance
is defined, for $f,f'\in\F$, by
\[
d\bigl(f,f'\bigr)= \bigl(\EXP\bigl[ \bigl(f(X)-f'(X)
\bigr)^2 \bigr] \bigr)^{1/2}
\]
and the $L_{\infty}$ distance is
\[
D\bigl(f,f'\bigr) = \sup_{x\in\X}\bigl\llvert
f(x)-f'(x)\bigr\rrvert.
\]
We also work with the (random) empirical quadratic distance
\[
d_{\mathbb{X}}\bigl(f,f'\bigr) = \Biggl(\frac{1}{n}\sum
_{i=1}^n \bigl(f(X_i)-f'(X_i)
\bigr)^2 \Biggr)^{1/2}.
\]
Denote by $f^*$ a function with minimal expectation
\begin{eqnarray*}
&&f^*=\mathop{\argmin}_{f \in\F} m_f.
\end{eqnarray*}

Next, we present two results that bound the excess risk
$m_{\hf}-m_{f^*}$ of the minimizer $\hf$ of Catoni's risk estimate in
terms of metric properties of the class $\F$.
The first result involves a combination of terms involving the $\gamma_2$
and $\gamma_1$
functionals under the metrics $d$ and $D$ while the second is in terms
of quantiles of $\gamma_2$ under the empirical metric $d_{\mathbb{X}}$.


\begin{theorem}
\label{teomain}
Let $\F$ be a class of nonnegative functions defined on
a set $\X$ and let $X,X_1,\ldots,X_n$ be i.i.d. random
variables taking values in $\X$.
Assume that there exists $v>0$
such that $\sup_{f\in\F} \operatorname{Var} (f(X) )\le v$.
Let $\delta\in(0,1/3)$.
Suppose that $\wh{f}$ is selected from $\F$ by minimizing Catoni's
mean estimator with parameter $\alpha$.
Then there exists a universal constant $L$ such that, under the condition
\[
6 \biggl(\alpha v+\frac{2\log(\delta^{-1})}{n\alpha} \biggr)+ L\log
\bigl(2\delta^{-1}
\bigr) \biggl(\frac{\gamma_2(\F,d)}{\sqrt{n}} +\frac{\gamma_1(\F,D)}{n}
\biggr) \le
\frac{1}{\alpha},
\]
%
with probability at least $1-3\delta$, the risk of $\wh{f}$
satisfies
\[
m_{\hf}-m_{f^*} \le6 \biggl(\alpha v+\frac{2\log(\delta
^{-1})}{n\alpha}
\biggr)+L\log\bigl(2\delta^{-1}\bigr) \biggl(\frac{\gamma
_2(\F,d)}{\sqrt{n}} +
\frac{\gamma_1(\F,D)}{n} \biggr).
\]
\end{theorem}

\begin{theorem}
\label{teoalternative}
Assume the hypotheses of Theorem \ref{teomain}. We denote by\break $
\operatorname{diam}_d(\F) $ the diameter of the class $\F$ under the
distance $d$.
Set $ \Gamma_{\delta} $ such that $ \mathbb{P} \{\gamma_2(\F
,d_{\mathbb{X}}) > \Gamma_{\delta} \} \le\frac{\delta}{8}$.
Then there exists a universal constant $K$ such that, under the condition
\[
6 \biggl(\alpha v+\frac{2\log(\delta^{-1})}{n\alpha} \biggr)+ K\max
\bigl(\Gamma_{\delta},
\operatorname{diam}_d(\F)\bigr)\sqrt{\frac{\log
(\sfrac{8}{\delta})}{n}} \le
\frac{1}{\alpha},
\]
%
with probability at least $1-3\delta$, the risk of $\wh{f}$
satisfies
\[
m_{\hf}-m_{f^*} \le6 \biggl(\alpha v+\frac{2\log(\delta
^{-1})}{n\alpha}
\biggr)+K\max\bigl(\Gamma_{\delta},\operatorname{diam}_d(\F)
\bigr)\sqrt{\frac{\log(\sfrac{8}{\delta})}{n}}.
\]
\end{theorem}

In both theorems above, the choice of $\alpha$ only influences the
term $\alpha v+2\log(\delta^{-1})/(n\alpha)$. By taking
$\alpha=\sqrt{2\log(\delta^{-1})/(nv)}$, this term equals
\[
2\sqrt{\frac{2v\log(\delta^{-1})}{n}}.
\]
For example, in that case, the condition in Theorem \ref{teomain}
reduces to
\[
12\sqrt{\frac{2v\log(\delta^{-1})}{n}}+L\log\bigl(\delta^{-1}\bigr)
\biggl(
\frac{\gamma_2(\F,d)}{\sqrt{n}} +\frac{\gamma_1(\F,D)}{n} \biggr) \le
\sqrt{\frac{nv}{2\log
(\delta^{-1})}}.
\]
This holds for sufficiently large values of $n$.
This choice has the disadvantage that the estimator depends on the
confidence level (i.e., on the value of $\delta$).
By taking $\alpha=\sqrt{2/(nv)}$, independently of $\delta$, one
obtains the slightly worse term
\[
\sqrt{\frac{2v}{n}}\bigl(1+\log\bigl(\delta^{-1}\bigr)\bigr).
\]
%
Observe that the main term in the second part of the bound of Theorem
\ref{teomain} is
\[
L\log\bigl(\delta^{-1}\bigr)\frac{\gamma_2(\F,d)}{\sqrt{n}}
\]
which is comparable to the bound (\ref{eqclassical}) obtained under
the strong condition of $f(X)$ being uniformly bounded.
All other terms are of smaller order.
Note that
this part of the bound depends on the ``weak'' distribution-dependent
$L_2(P)$ metric $d$. The quantity $\gamma_1(\F,D) \ge\gamma_2(\F
,d)$ also enters the
bound of Theorem \ref{teomain} though only multiplied by $1/n$. The
presence of this term
requires that $\F$ be bounded in the $L_{\infty}$ distance $D$ which
limits the usefulness of the bound.
In Section~\ref{secapps}, we illustrate the bounds on two applications
to regression and $k$-means clustering. In these applications, in spite
of the presence of heavy tails, the covering numbers under the distance $D$
may be bounded in a meaningful way.
Note that no such bound can hold
for ``ordinary'' empirical risk minimization that minimizes the
usual empirical means $(1/n)\sum_{i=1}^n f(X_i)$ because of the
poor performance of empirical averages in the presence of heavy
tails.

The main merit of the bound of Theorem \ref{teoalternative} is that it
does not require that the class $\F$ has a finite diameter under the
supremum norm. Instead, the quantiles of $\gamma_2(\F,d_{\mathbb{X}})$
enter the picture. In Section~\ref{secapps}, we show through
the example of $L_2$ regression how these quantiles may be estimated.

\section{Proofs}
\label{secproofs}

The proofs of Theorems \ref{teomain} and \ref{teoalternative}
are based on showing that the excess risk can be bounded as soon
as the supremum of the empirical process
$\{X_f(\mu): f\in\F\}$ is bounded for any
fixed $\mu\in\R$, where for any $f\in\F$ and $\mu\in\R$, we define
$X_f(\mu)=\hr_f(\mu)-\br_f(\mu)$ with
\[
\ol{r}_f (\mu) = \frac{1}{\alpha} \mathbb{E} \bigl[\phi\bigl(
\alpha\bigl(f(X)-\mu\bigr)\bigr) \bigr]
\]
and
\[
\wh{r}_f(\mu) = \frac{1}{n\alpha} \sum
_{i=1}^n\phi\bigl(\alpha\bigl(f(X_i)-
\mu\bigr)\bigr).
\]
The two theorems differ in the way the supremum of this empirical
process is bounded.

Let $A_{\alpha}(\delta)=\alpha v+2\log(\delta^{-1})/(n\alpha)$.

Once again, we may assume, essentially without loss of generality, that the
minimum exists. In case of multiple minimizers, we may choose one
arbitrarily.
The main result in \cite{Cat10} states that for any $\delta>0$ such that
$\alpha^2v + 2\log(\delta^{-1})/n \le1$,
with probability at least $1-2\delta$,
%
\begin{equation}
\label{eq-cat} \llvert\wh{\mu}_{f^*}- m_{f^*}\rrvert\le
A_{\alpha
}(\delta).
\end{equation}
%
Let $\Omega_{f^*}(\delta) $ be the event on which inequality (\ref
{eq-cat}) holds. By definition,\break $\mathbb{P} \{\Omega
_{f^*}(\delta) \} \ge1-2\delta$.

\subsection{A deterministic version of \texorpdfstring{$\wh{\mu}_f$}{$widehat{mu}_f$}}
We begin with a variant of the argument of Catoni \cite{Cat10}.
It involves a deterministic version $\ol{\mu}_f$ of the estimator defined,
for each $f\in\F$, as the unique solution of the equation
$\ol{r}_f (\mu) = 0$.

In Lemma \ref{lembr} below, we show that $\ol{\mu}_f$ is
in a small (deterministic) interval centered at $m_f$.
For any $f\in\F$, $\mu\in\R$, and $\ep\ge0$, define
\begin{eqnarray*}
B_f^+(\mu,\ep)&=&(m_f-\mu)+\frac{\alpha}{2}(m_f-
\mu)^2+\frac
{\alpha}{2}v+\ep,
\\
B_f^-(\mu,\ep)&=&(m_f-\mu)-\frac{\alpha}{2}(m_f-
\mu)^2-\frac
{\alpha}{2}v-\ep
\end{eqnarray*}
and let
\[
\mu_f^+(\ep)=m_f+\alpha v+2\ep, \qquad
\mu_f^-(\ep)=m_f-\alpha v-2\ep.
\]
As a function of $\mu$, $B_f^+(\mu,\ep)$ is a quadratic polynomial
such that $\mu_f^+(\ep)$ is an upper bound of the smallest root of
$B_f^+(\mu,\ep)$. Similarly, $\mu_f^-(\ep)$ is a lower bound of the
largest root of $B_f^-(\mu,\ep)$. Implicitly, we assumed that these
roots always exist. This is not always the case but a simple condition
on $\alpha$ guarantees that these roots exists.
In particular,
$ 1-\alpha^2v-2\alpha\ep\ge0 $ guarantees that $ B_f^+(\mu,\ep)=0
$ and $ B_f^-(\mu,\ep)=0 $ have at least one solution. This condition
will always
be satisfied by our choice of $\epsilon$ and $\alpha$.



%

Still following the ideas of \cite{Cat10}, the next lemma bounds $\ol
{r}_f (\mu)$
by the quadratic polynomials $B^+$ and $B^-$. The lemma will help us
compare the zero of $\ol{r}_f (\mu)$ to the zeros of these quadratic
functions.

\begin{lem}
\label{lembr}
For any fixed $f\in\F$ and $\mu\in\R$,
%
\begin{equation}
\label{eq-det} B_f^-(\mu,0)\le\br_f(\mu) \le
B_f^+(\mu,0),
\end{equation}
and, therefore, $m_f-\alpha v\le\ol{\mu}_f \le m_f+\alpha v$.
In particular,
\[
\label{eq-br} B_{\hf}^-(\mu,0)\le\br_{\hf}(\mu) \le
B_{\hf}^+(\mu,0).
\]
For any $\mu$ and $\ep$, such that $\br_{\hf}(\mu) \le\ep$, if $
1-\alpha^2v-2\alpha\ep\ge0 $, then
%
\begin{equation}
\label{eq-mf} m_{\hf} \le\mu+\alpha v+ 2\ep.
\end{equation}
\end{lem}

\begin{pf}
%
\begin{figure}[b]

\includegraphics{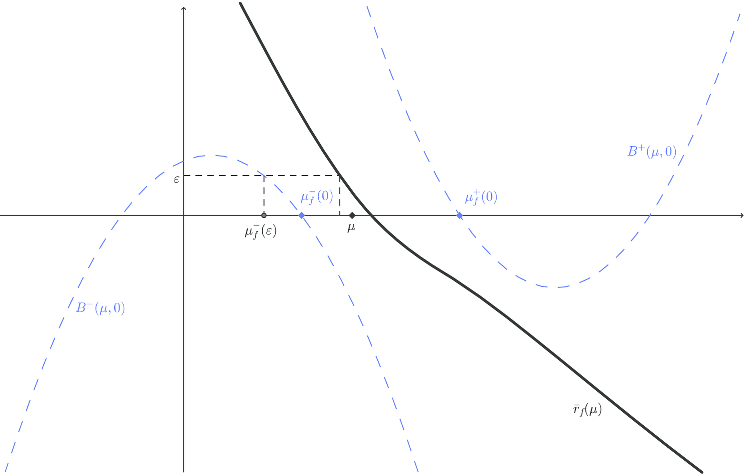}

\caption{Representation of $ \br_f(\mu) $ and the quadratic
functions $ B_f^-(\mu,0)$ and $B_f^+(\mu,0)$.
$ \br_f(\mu) $ is squeezed between $ B_f^-(\mu,0)$ and $B_f^+(\mu,0)
$. In particular at $ \mu^+_f(0) $ [resp., $ \mu^-_f(0)$], $\br
_f(\mu)$ is nonpositive (resp., nonnegative). Any $\mu$ such that $\br_f(\mu
) \le\ep$ is above $\mu_f^{-}(\ep)$.}\label{figBf}
\end{figure}
%
Writing $Y$ for $\alpha(f(X)-\mu)$ and using the fact that
$\phi(x)\le\log(1+x+x^2/2)$ for all $x\in\R$,
\begin{eqnarray*}
\exp\bigl(\alpha\ol{r}_f (\mu) \bigr) & \le& \exp\biggl( \EXP
\biggl[\log\biggl(1+Y+\frac{Y^2}{2}\biggr) \biggr] \biggr)
\\
& \le& \EXP\biggl[1+Y+\frac{Y^2}{2} \biggr]
\\
& \le& 1+\alpha(m_f-\mu) + \frac{\alpha^2}{2} \bigl[v+(m_f-
\mu)^2 \bigr]
 \le \exp\bigl(\alpha B_f^+ (\mu,0) \bigr).
\end{eqnarray*}
Thus, we have $\ol{r}_f (\mu) - B_f^+(\mu,0) \le0$ (see Figure~\ref
{figBf}). Since this last inequality is true for any $f$,
$\sup_f(\ol{r}_{f} (\mu) - B_{f}^+(\mu,0)) \le0$ and the second
inequality of (\ref{eq-det}) is proved. The second statement of the
lemma may be proved by a similar argument.

If $\br_{\hf}(\mu) \le\ep$, then $B_{\hf}^-(\mu,0)\le\ep$
which is equivalent to $B_{\hf}^-(\mu,\ep)\le0$. If $ 1-\alpha
^2v-2\alpha\ep\ge0 $ then a solution of $B_{\hf}^-(\mu,\ep)= 0$
exists and since $\br_{\hf}(\mu)$ is a nonincreasing function, $\mu
$ is above the largest of these two solutions. This implies $\mu
^{-}_{\hf}(\ep)\le\mu$ which gives inequality (\ref{eq-mf}) (see
Figure~\ref{figBf}).
\end{pf}

Inequality (\ref{eq-mf}) is the key tool to ensure that the risk
$m_{\hf}$ of the minimizer $\hf$ can be upper bounded as soon as
$\overline{r}_{\hf}$ is.
It remains to find the smallest $\mu$ and $\ep$ such that $\overline
{r}_{f}(\mu)$ is bounded uniformly on $\F$.

\subsection{Bounding the excess risk in terms of the supremum of an
empirical process}


The key to all proofs is that we link the excess risk to the supremum
of the
empirical process $X_f(\mu)=\hr_f(\mu)-\br_f(\mu)$ as $f$ ranges
through $\F$ for a suitably
chosen value of $\mu$. For fixed $\mu\in\R$ and $\delta\in(0,1)$,
define the $1-\delta$ quantile of $\sup_{f\in\F}\llvert X_f(\mu
)-X_{f^*}(\mu)\rrvert $
by $Q(\mu,\delta)$, that is, the infimum of all positive
numbers $q$ such that
\[
\PROB\Bigl\{\sup_{f\in\F}\bigl\llvert X_f(
\mu)-X_{f^*}(\mu)\bigr\rrvert\le q \Bigr\} \ge1-\delta.
\]
First, we need a few simple facts summarized in the next lemma.

\begin{lem}
\label{lemmu0}
Let $\mu_0=m_{f^*}+A_{\alpha}(\delta)$. Then on the event $\Omega
_{f^*}(\delta)$, the following inequalities hold:
\begin{longlist}[2.]
\item[1.]$ \hr_{\hf}(\mu_0)\le0 $;
\item[2.]$\br_{f^*}(\mu_0) \le0$;
\item[3.]$-\hr_{f^*}(\mu_0) \le2A_{\alpha}(\delta)$.
\end{longlist}
\end{lem}

\begin{pf}
We prove each inequality separately.
\begin{longlist}[3.]
\item[1.] First, note that on $ \Omega_{f^*}(\delta)$ inequality (\ref
{eq-cat}) holds, and
we have $\wh{\mu}_{\hf} \le\wh{\mu}_{f^*}\le\mu_0 $. Since $\hr
_{\hf}$ is a nonincreasing function of $\mu$,
$\hr_{\hf}(\mu_0)\le\hr_{\hf}(\wh{\mu}_{\hf})=0$.

\item[2.] By (\ref{eq-det}), $\wb{\mu}_{f^*} \le m_{f^*}+\alpha v\le
m_{f^*}+\alpha v +2\log(\delta^{-1})/(n \alpha)=\mu_0$. Since $\br
_{f^*} $ is a nonincreasing function, $\br_{f^*}(\mu_0) \le\br
_{f^*}(\wb{\mu}_{f^*}) = 0$.

\item[3.]$\hr_{f^*}$ is a $1$-Lipschitz function and, therefore,
\begin{eqnarray*}
\bigl\llvert\hr_{f^*}(\mu_0)\bigr\rrvert&=&\bigl\llvert
\hr_{f^*}(\wh{\mu}_{f^*})-\hr_{f^*}(\mu
_0)\bigr\rrvert \le \llvert\wh{\mu}_{f^*}-
\mu_0\rrvert
\\
&\le&\llvert\wh{\mu}_{f^*}-m_{f^*}\rrvert+\llvert
m_{f^*}-\mu_0\rrvert
\\
&\le&2A_{\alpha}(\delta)
\end{eqnarray*}
which gives $-\hr_{f^*}(\mu_0) \le2A_{\alpha}(\delta)$.\quad\qed
\end{longlist}\noqed
\end{pf}

We will use Lemma \ref{lembr} with $\mu_0$ introduced in Lemma \ref
{lemmu0}. Recall that $\mathbb{P} \{\Omega_{f^*}(\delta)
\} \ge1-2\delta$.

With the notation introduced above, we see that with probability at
least $1-\delta$,
\begin{eqnarray*}
\br_{\hf}(\mu_0) &\le& \hr_{\hf}(
\mu_0)+\br_{f^*}(\mu_0)-\hr_{f^*}(
\mu_0)
\\
&&{} +\bigl\llvert\br_{\hf}(\mu_0)-
\hr_{\hf}(\mu_0)-\br_{f^*}(\mu_0)+
\hr_{f^*}(\mu_0)\bigr\rrvert
\\
&\le& \hr_{\hf}(\mu_0)+\br_{f^*}(
\mu_0)-\hr_{f^*}(\mu_0)
\\
&&{} +\sup
_{f\in\F}\bigl\llvert\br_{f}(\mu_0)-
\hr_{f}(\mu_0)-\br_{f^*}(\mu_0)+
\hr_{f^*}(\mu_0)\bigr\rrvert
\\
&\le& \hr_{\hf}(\mu_0)+\br_{f^*}(
\mu_0)-\hr_{f^*}(\mu_0) +Q(\mu_0,
\delta).
\end{eqnarray*}
This inequality, together with Lemma \ref{lemmu0}, implies that, with
probability at least $1-3\delta$,
\[
\br_{\hf}(\mu_0) \le2A_{\alpha}(\delta)+Q(
\mu_0,\delta).
\]
Now using Lemma \ref{lembr} with $ \ep= 2A_{\alpha}(\delta)+Q(\mu
_0,\delta) $ and under the condition $ 1-\alpha^2v-4\alpha A_{\alpha
}(\delta)-2\alpha Q(\mu_0,\delta) \ge0 $, we have
%
\begin{eqnarray}
\label{eqprincipal} m_{\hf} - m_{f^*}&\le& \alpha v+5
A_{\alpha}(\delta)+2 Q(\mu_0,\delta)
\nonumber\\[-8pt]\\[-8pt]\nonumber
&\le& 6 \biggl(\alpha v +\frac{2\log(\delta^{-1})}{n \alpha} \biggr) +
2 Q(\mu_0,
\delta),
\end{eqnarray}
with probability at least $1-3\delta$.
The condition $ 1-\alpha^2v-4\alpha A_{\alpha}(\delta)-2\alpha
Q(\mu_0,\delta) \ge0$ is satisfied whenever
\[
6 \biggl(\alpha v +\frac{2\log(\delta^{-1})}{n \alpha} \biggr)+2 Q(\mu
_0,\delta) \le
\frac{1}{\alpha}
\]
holds.

\subsection{Bounding the supremum of the empirical process}
\label{subsecchain}

Theorems \ref{teomain} and \ref{teoalternative} both follow
from (\ref{eqprincipal}) by two different ways of bounding the quantile
$Q(\mu,\delta)$ of $\sup_{f\in\F}\llvert X_f(\mu)-X_{f^*}(\mu
)\rrvert $.
Here, we present these two inequalities. Both of them use
basic results of ``generic chaining''; see Talagrand \cite{talagrand}.
Theorem \ref{teomain} follows from (\ref{eqprincipal}) and the
next inequality.

\begin{prop}
\label{prop-controlrf}
Let $\mu\in\R$ and $\alpha>0$. There exists a universal constant
$L$ such that for any $\delta\in(0,1)$,
\[
Q(\mu,\delta) \le L\log\bigl(2\delta^{-1}\bigr) \biggl(
\frac{\gamma_2(\F,d)}{\sqrt{n}} +\frac{\gamma_1(\F,D)}{n} \biggr).
\]
\end{prop}

The proof is an immediate consequence of Theorem \ref{teochaining}
and (\ref{eq-orlic}) in the
\hyperref[app]{Appendix} and the following lemma.

\begin{lemma}
\label{lem-concentration1}
For any $\mu\in\R$, $ \alpha>0 $,
$f,f' \in\F$, and $t>0$,
\[
\mathbb{P} \bigl\{\bigl\llvert X_f(\mu)-X_{f'}(\mu)\bigr
\rrvert>t \bigr\} \le2\exp\biggl(-\frac{nt^2}{2(d(f,f')^2+\sklafrac
{2D(f,f')t}{3})} \biggr),
\]
where the distances $d,D$ are defined at the beginning of Section~\ref
{secresults}.
\end{lemma}
%

\begin{pf}
Observe that $n(X_f(\mu)- X_{f'}(\mu))$ is the sum of the independent
zero-mean
random variables
\begin{eqnarray*}
C_i\bigl(f,f'\bigr) & = & \frac{1}{\alpha}\phi
\bigl(\alpha\bigl(f(X_i)-\mu\bigr)\bigr)-\frac{1}{\alpha}\phi\bigl(
\alpha\bigl(f'(X_i)-\mu\bigr)\bigr)
\\
& &{} - \biggl[\frac{1}{\alpha}\mathbb{E} \bigl[\phi\bigl(\alpha\bigl
(f(X)-\mu
\bigr)\bigr) \bigr]-\frac
{1}{\alpha}\mathbb{E} \bigl[\phi\bigl(\alpha
\bigl(f'(X)-\mu\bigr)\bigr) \bigr] \biggr].
\end{eqnarray*}
Note that since the truncation function $\phi$ is 1-Lipschitz, we have
$C_i(f,f')\leq2D (f,f') $. Also,
\[
\sum_{i=1}^{n} \EXP\bigl[C_i
\bigl(f,f'\bigr)^2\bigr] \leq\sum
_{i=1}^{n}\EXP\bigl[ \bigl(\bigl(f(X_i)-
\mu\bigr)-\bigl(f'(X_i)-\mu\bigr) \bigr)^2
\bigr]=n d\bigl(f,f'\bigr)^2.
\]
The lemma follows from Bernstein's inequality [see, e.g., \cite{gabor},
equation~(2.10)].
\end{pf}

Similarly, Theorem \ref{teoalternative} is implied by (\ref
{eqprincipal}) and
the following. Recall the notation of Theorem \ref{teoalternative}.

\begin{teo}
Let $\mu\in\R$, $\alpha>0$, and $\delta\in(0,1/3)$.
There exists a universal constant $K$ such that
\[
Q(\mu,\delta) \le K\max\bigl(\Gamma_{\delta},\operatorname{diam}_d(
\F)\bigr)\sqrt{\frac
{\log(\sfrac{8}{\delta})}{n}} .
\]
\end{teo}
%

\begin{pf}
Assume $ \Gamma_{\delta} \ge\operatorname{diam}_d(\F) $. The proof
is based on a standard symmetrization argument.
Let $(X'_1,\dots,X'_n)$ be independent copies of $(X_1,\dots,X_n)$ and
define
\[
Z_i(f)= \frac{1}{n\alpha}\phi\bigl(\alpha\bigl(f(X_i)-
\mu\bigr)\bigr)-\frac
{1}{n\alpha}\phi\bigl(\alpha\bigl(f\bigl(X'_i
\bigr)-\mu\bigr)\bigr).
\]
Introduce also independent Rademacher random variables
$(\ep_1,\dots,\ep_n)$. For any $f\in\F$, denote by $Z(f) = \sum
_{i=1}^{n}\ep_iZ_i(f)$. Then by Hoeffding's inequality, for
all
$f,g \in\F$ and for every $t>0$,
%
\begin{equation}
\label{eqhoeffding} \mathbb{P}_{(\ep_1,\dots,\ep_n)} \bigl\{\bigl
\llvert Z(f)-Z(g)\bigr
\rrvert>t \bigr\} \le2\exp\biggl(-\frac
{t^2}{2d_{\mathbb{X},\mathbb{X}'}(f,g)^2} \biggr),
\end{equation}
where
$\PROB_{(\ep_1,\dots,\ep_n)}$ denotes probability with respect to the
Rademacher\break variables only (i.e., conditional on the $X_i$ and $X_i'$)
and
$d_{\mathbb{X},\mathbb{X}'}(f,g) =\break \sqrt{\sum
_{i=1}^n(Z_i(f)-Z_i(g))^2}$ is a random distance.
Using (\ref{eqcharac}) in the \hyperref[app]{Appendix} with distance $ d_{\mathbb
{X},\mathbb{X}'} $ and (\ref{eqhoeffding}),
we get that, for all $\lambda>0$,
%
\begin{eqnarray}
\label{eqsubgauss}
&& \mathbb{E}_{(\ep_1,\dots,\ep_n)} \Biggl[\exp\Biggl
(\lambda\sup
_{f\in\F}\Biggl\llvert\sum_{i=1}^n
\ep_i\bigl[Z_i(f)-Z_i\bigl(f^*\bigr)\bigr]
\Biggr\rrvert\Biggr) \Biggr]
\nonumber\\[-8pt]\\[-8pt]\nonumber
&&\qquad \le2\exp\bigl(\lambda^2
L^2 \gamma_2(\F,d_{\mathbb{X},\mathbb
{X}'})^2/4 \bigr),
\end{eqnarray}
where $L$ is a universal constant from Proposition \ref{cororlic}.
Observe that since $x \mapsto\phi( x)$ is Lipschitz with constant $1$,
\begin{eqnarray*}
&& d_{\mathbb{X},\mathbb{X}'}(f,g)
\\
&&\qquad = \Biggl(\frac{1}{n^2\alpha^2}\sum_{i=1}^n
\bigl(\phi\bigl(\alpha\bigl(f(X_i)-\mu\bigr) \bigr)-\phi\bigl(
\alpha\bigl(f\bigl(X'_i\bigr)-\mu\bigr) \bigr)
\\
&&\quad\qquad{}  -\phi
\bigl(\alpha\bigl(g(X_i)-\mu\bigr) \bigr)+\phi\bigl(\alpha\bigl(g
\bigl(X'_i\bigr)-\mu\bigr) \bigr) \bigr)^2
\Biggr)^{1/2}
\\
&&\qquad \le\frac{1}{\sqrt{n}} \Biggl(\frac{1}{n}\sum
_{i=1}^n \bigl(f(X_i)-g(X_i)
\bigr)^2 \Biggr)^{1/2}+\frac{1}{\sqrt{n}} \Biggl(
\frac
{1}{n}\sum_{i=1}^n \bigl(f
\bigl(X'_i\bigr)-g\bigl(X'_i
\bigr) \bigr)^2 \Biggr)^{1/2}.
\end{eqnarray*}
This implies
\[
\gamma_2(\F,d_{\mathbb{X},\mathbb{X}'}) \le\frac{1}{\sqrt
{n}} \bigl(
\gamma_2(\F,d_{\mathbb{X}})+\gamma_2(
\F,d_{\mathbb
{X}'}) \bigr).
\]
Combining this with (\ref{eqsubgauss}), we obtain
\begin{eqnarray*}
&& \mathbb{P} \Bigl\{\sup_{f\in\F} \bigl\llvert Z(f)-Z\bigl(f^*
\bigr)\bigr\rrvert\ge t \Bigr\}
\\
&&\qquad\le\mathbb{P} \Bigl\{\sup_{f\in\F} \bigl\llvert Z(f)-Z
\bigl(f^*\bigr)\bigr\rrvert\ge t \bigl\llvert\gamma_2(
\F,d_{\mathbb{X}}) \le\Gamma_{\delta
}\mbox{ and } \gamma_2(
\F,d_{\mathbb{X}'})\le\Gamma_{\delta} \Bigr\}
\\
&&\quad\qquad{}+2\mathbb{P} \bigl\{
\gamma_2(\F,d_{\mathbb{X}}) > \Gamma_{\delta
} \bigr\}
\\
&&\qquad{} \le
\mathbb{E}_{\mathbb{X},\mathbb{X}'} \bigl[\mathbb{E}_{(\ep
_1,\dots,\ep_n)} \bigl[e^{\lambda\sup_{f\in\F}\llvert \sum
_{i=1}^n \ep_i[Z_i(f)-Z_i(f^*)]\rrvert }
\bigr]\bigr\rrvert\gamma_2(\F,d_{\mathbb{X}}) \le
\Gamma_{\delta}\mbox{ and}
\\
&&\quad\qquad{}  \gamma_2(\F,d_{\mathbb{X}'})\le
\Gamma_{\delta} \bigr]e^{-\lambda
t}
\\
&&{}\qquad\quad +\frac{\delta}{4} \qquad\mbox{(by the definition of $\Gamma_\delta$)}
\\
&&\qquad \le 2\exp\biggl(\frac{\lambda^2 L^2}{n} \Gamma_{\delta}^2-
\lambda t \biggr)+\frac{\delta}{4}.
\end{eqnarray*}
%
Optimization in $\lambda$ with $t= 2L\Gamma_{\delta}\sqrt{\log
(8/\delta)/n}$ gives
\[
\mathbb{P} \Bigl\{\sup_{f\in\F} \bigl\llvert Z(f)-Z\bigl(f^*
\bigr)\bigr\rrvert\ge t \Bigr\} \le\frac{\delta}{2}.
\]
A standard symmetrization inequality of tail probabilities of
empirical processes (see, e.g., \cite{Geer00}, Lemma 3.3) guarantees that
\[
\mathbb{P} \Bigl\{\sup_{f\in\F} \bigl\llvert X_f(
\mu)-X_{f^*}(\mu)\bigr\rrvert\ge2t \Bigr\}\le2\mathbb{P} \Bigl\{\sup
_{f\in\F} \bigl\llvert Z(f)-Z\bigl(f^*\bigr)\bigr\rrvert\ge t
\Bigr\}
\]
as long as for any $f \in\F$, $\mathbb{P} \{\llvert X_f(\mu
)-X_{f^*}(\mu)\rrvert \ge t \}
\le\frac{1}{2}$. Recall that $X_f(\mu)-X_{f^*}(\mu)$ is a zero-mean
random variable. Then by Chebyshev's inequality, it suffices to have $
t \ge\sqrt{2}\operatorname{diam}_d(\F)/\sqrt{n}$. Indeed,
\begin{eqnarray*}
&& \frac{\operatorname{Var} (X_f(\mu)-X_{f'}(\mu) )}{t^2}
\\
&&\qquad  \le \frac
{\operatorname{Var} (\sklfrac{1}{\alpha}\phi(\alpha
(f(X)-\mu))-\sklfrac{1}{\alpha}\phi(\alpha(f^*(X)-\mu)) )}{nt^2}
\\
&&\qquad \le \frac{\mathbb{E} [(f(X)-f^*(X))^2 ]}{nt^2}
\\
&&\qquad \le \frac{\operatorname{diam}_d(\F)^2}{nt^2}.
\end{eqnarray*}
Without loss of generality, we can assume $L \ge1$. Since for any
choice of $\delta<\frac{1}{3} $, $ \sqrt{\log(\frac{8}{\delta})}
> \sqrt{2} $ we have $L\Gamma_{\delta}\sqrt{\log(\frac{8}{\delta
})} \ge
\operatorname{diam}_d(\F) \sqrt{2}$. Thus,
\[
\mathbb{P} \biggl\{\sup_{f\in\F} \bigl\llvert X_f(
\mu)-X_{f^*}(\mu)\bigr\rrvert\ge2L\Gamma_{\delta}\sqrt{
\frac{\log(\sfrac
{8}{\delta
})}{n}} \biggr\} \le\delta
\]
as desired. Now, if $\Gamma_{\delta} < \operatorname{diam}_d(\F)$,
$ \mathbb{P} \{\gamma_2(\F,d_{\mathbb{X}}) > \operatorname
{diam}_d(\F) \} \le\frac{\delta}{8} $ and the same argument
holds for
$\operatorname{diam}_d(\F)$ instead of $ \Gamma_{\delta} $. This
completes the proof.
\end{pf}

\section{Applications}
\label{secapps}

In this section, we describe two applications of Theorems \ref
{teomain} and \ref{teoalternative}
to simple statistical learning problems. The first is a regression
estimation problem in which
we distinguish between $L_1$ and $L_2$ risks
and the second is $k$-means clustering.

\subsection{Empirical risk minimization for regression}

\subsubsection{$L_1$ regression}
\label{secL1regr}
Let $(Z_1,Y_1),\ldots,(Z_n,Y_n)$ be independent identically
distributed random variables
taking values in $\Z\times\R$ where $\Z$ is a bounded subset of
(say) $\R^m$.
Suppose $\G$ is a class of functions $\Z\to\R$
bounded in the $L_{\infty}$ norm, that is,
$\sup_{g\in\G} \sup_{z\in\Z}\llvert g(z)\rrvert <\infty$.
We denote by $\Delta$ the diameter of $\G$ under the distance induced
by this norm.
First, we consider the setup when
the \textit{risk} of each $g\in\G$ is defined by the $L_1$ loss
\[
R(g) = \EXP\bigl\llvert g(Z)-Y\bigr\rrvert,
\]
where the pair $(Z,Y)$ has the same distribution of the $(Z_i,Y_i)$
and is independent of them. Let $g^*= \argmin_{g\in\G} R(g)$
be a minimizer of the risk (which, without loss of generality,
is assumed to exist).
The statistical learning problem we consider here consists of choosing
a function $\wh{g}$ from the class $\G$ that has a risk $R(\wh{g})$
not much larger than $R(g^*)$.

The standard procedure is to pick $\wh{g}$ by minimizing the
empirical risk  $(1/n) \sum_{i=1} \llvert g(Z_i)-Y_i\rrvert $
over $g\in\G$.
However, if the response variable $Y$ is unbounded and may have a heavy tail,
ordinary empirical risk minimization may fail to provide a good
predictor of $Y$ as
the empirical risk is an unreliable estimate of the true risk.

Here, we propose choosing $\wh{g}$ by minimizing Catoni's estimate.
To this end, we only need to assume that the second moment of $Y$
is bounded by a known constant. More precisely, assume that
$\EXP Y^2 \le\sigma^2$ for some $\sigma>0$. Then
$\sup_{g\in\G} \operatorname{Var} (\llvert g(Z)-Y\rrvert )\le2\sigma
^2 + 2\sup_{g\in\G}
\sup_{z\in\Z}\llvert g(z)\rrvert ^2
\stackrel{\mathrm{def}}{=}v$ is a known and finite constant.

Now for all $g\in\G$ and $\mu\in\R$, define
\[
\wh{r}_g(\mu) = \frac{1}{n\alpha} \sum
_{i=1}^n\phi\bigl(\alpha\bigl(\bigl\llvert
g(X_i)-Y_i\bigr\rrvert-\mu\bigr)\bigr),
\]
where $\phi$ is the truncation function
defined in (\ref{eqtrunc}).
Define $\wh{R}(g)$ as the unique value for which
$\wh{r}_g(\wh{R}(g))=0$. The empirical risk minimizer based on Catoni's
risk estimate is then
\[
\wh{g}= \argmin_{g\in\G} \wh{R}(g).
\]
By Theorem \ref{teomain}, the performance of $\wh{g}$ may be
bounded in terms of covering numbers of the class of functions
$\F=\{f(z,y)=\llvert g(z)-y\rrvert : g\in\G\}$ based on the distance
\[
D\bigl(f,f'\bigr)=\sup_{z\in\Z,y\in\R} \bigl\llvert\bigl
\llvert g(z)-y\bigr\rrvert-\bigl\llvert g'(z)-y\bigr\rrvert\bigr
\rrvert\le\sup_{z\in\Z} \bigl\llvert g(z) -g'(z)
\bigr\rrvert.
\]
Thus, the covering numbers of $\F$ under the distance $D$ may
be bounded in terms of the covering numbers of $\G$ under the
$L_{\infty}$ distance. Denoting by $N_{d}(A,\epsilon)$ the
$\epsilon$-covering
number of a set $A$ under the metric $d$,
we obtain the following.

\begin{cor}
Consider the setup described above. We assume\break $\int_0^\Delta\log
N_{\infty}(\G,\epsilon)\,d\epsilon< \infty$. Let $n\in\N$, $\delta
\in(0,1/3)$ and $\alpha=\sqrt{2\log(\delta^{-1})/(nv)}$. There
exists an integer $N_0$ and a universal
constant $C$ such that, for all $n \ge N_0$, with probability
at least $1-3\delta$,
\begin{eqnarray*}
&& R(\wh{g}) - R\bigl(g^*\bigr)
\\
&&\qquad \le12\sqrt{\frac{2v\log(\delta^{-1})}{n}}+
C\log\bigl(2
\delta^{-1}\bigr) \biggl(\frac{1}{\sqrt{n}} \int_0^\Delta
\sqrt{\log N_{d}(\G,\epsilon)}\,d\epsilon+O \biggl(\frac{1}{n}
\biggr) \biggr).
\end{eqnarray*}
\end{cor}

\begin{pf}
Clearly, if two distances $d_1$ and $d_2$ satisfy $d_1 \le d_2$, then
$ \gamma_1(\F,d_1)\le\gamma_1(\F,d_2) $. Thus, \mbox{$ \gamma_1(\F,D)
\le\gamma_1(\G,\llVert \cdot \rrVert _{\infty}) \le L\int_0^\Delta\log
N_{\infty}(\G,\epsilon)\,d\epsilon< \infty$} [see~(\ref{eqdudley})]
and $\gamma_1(\F,D)/n = O (1/n ) $. The condition
\begin{eqnarray*}
&& 12\sqrt{\frac{2v\log(\delta^{-1})}{n}}+ C\log\bigl(2\delta^{-1}\bigr)
\biggl(
\frac{1}{\sqrt{n}} \int_0^\Delta\sqrt{\log
N_{d}(\G,\epsilon)}\,d\epsilon+O \biggl(\frac{1}{n} \biggr)
\biggr)
\\
&&\qquad  \le\sqrt{\frac{nv}{2\log(\delta^{-1})}}
\end{eqnarray*}
is satisfied for sufficiently large $n$.
Apply Theorem \ref{teomain}.
\end{pf}

Note that the bound essentially has the same form as
(\ref{eqdudley1}) but to apply (\ref{eqdudley1}) it is crucial
that the response variable $Y$ is bounded or at least has
sub-Gaussian tails. We get this under the weak assumption that
$Y$ has a bounded second moment (with a known upper bound).
The price we pay is that covering numbers under the distance
$d_{\mathbb{X}}$
are now replaced by covering numbers under the supremum norm.

\subsubsection{$L_2$ regression}
\label{secl2}

Here, we consider the same setup as in Section~\ref{secL1regr} but
now the
risk is measured by the $L_2$ loss.
The \textit{risk} of each $g\in\G$ is defined by the $L_2$ loss
\[
R(g) = \EXP\bigl(g(Z)-Y\bigr)^2.
\]
Note that Theorem \ref{teomain} is useless here as the difference $\llvert
R(g)-R(g')\rrvert $ is
not bounded by the $L_{\infty}$ distance
of $g$ and $g'$ anymore and the covering numbers of $\F$ under the
metric $D$ are infinite.
However, Theorem \ref{teoalternative} gives meaningful bounds.
Let $g^*= \argmin_{g\in\G} R(g)$ and again we choose $\wh{g}$ by
minimizing Catoni's estimate.

Here, we need to assume that
$\EXP Y^4 \le\sigma^2$ for some $\sigma>0$. Then\break
$\sup_{g\in\G} \operatorname{Var} ((g(Z)-Y)^2 )\le
8\sigma^2 + 8\sup_{g\in\G}
\sup_{z\in\Z}\llvert g(z)\rrvert ^4
\stackrel{\mathrm{def}}{=}v$ is a known and finite constant.

By Theorem \ref{teoalternative}, the performance of $\wh{g}$ may be
bounded in terms of covering numbers of the class of functions
$\F=\{f(z,y)=(g(z)-y)^2: g\in\G\}$ based on the distance
\[
d_{\mathbb{X}}\bigl(f,f'\bigr)= \Biggl(\frac{1}{n}\sum
_{i=1}^n \bigl(\bigl(g(Z_i)-Y_i
\bigr)^2 -\bigl(g'(Z_i)-Y_i
\bigr)^2 \bigr)^2 \Biggr)^{1/2}.
\]
Note that
\begin{eqnarray*}
\bigl\llvert\bigl(g(Z_i)-Y_i\bigr)^2-
\bigl(g'(Z_i)-Y_i\bigr)^2 \bigr
\rrvert&=&\bigl\llvert g(Z_i)-g'(Z_i)\bigr
\rrvert\bigl\llvert2Y_i-g(Z_i)-g'(Z_i)
\bigr\rrvert
\\
&\le& 2\bigl\llvert g(Z_i)-g'(Z_i)\bigr
\rrvert\bigl(\llvert Y_i\rrvert+\Delta\bigr)
\\
&\le& 2d_{\infty}\bigl(g,g'\bigr) \bigl(\llvert
Y_i\rrvert+\Delta\bigr),
\end{eqnarray*}
and, therefore,
\begin{eqnarray*}
d_{\mathbb{X}}\bigl(f,f'\bigr) &\le& 2 d_{\infty}
\bigl(g,g'\bigr)\sqrt{\frac{1}{n}\sum
_{i=1}^n\bigl(\llvert Y_i\rrvert+
\Delta\bigr)^2}
\\
& \le&2 \sqrt{2} d_{\infty}\bigl(g,g'\bigr)\sqrt{
\Delta^2+\frac{1}{n}\sum_{i=1}^nY_i^2}.
\end{eqnarray*}
By Chebyshev's inequality,
\[
\mathbb{P} \Biggl\{\frac{1}{n}\sum_{i=1}^nY_i^2
- \mathbb{E} \bigl[Y^2 \bigr]> t \Biggr\} \le\frac{\operatorname{Var}
(Y^2 )}{nt^2} \le
\frac
{\sigma^2}{nt^2}
\]
thus $ \frac{1}{n}\sum_{i=1}^nY_i^2 > \mathbb{E} [Y^2 ] +
\sqrt{8\sigma
^2/(n\delta)} $ with probability at most $\delta/8$ and
\[
d_{\mathbb{X}}\bigl(f,f'\bigr) > 2 \sqrt{2} d_{\infty}
\bigl(g,g'\bigr)\sqrt{\Delta^2+\mathbb{E}
\bigl[Y^2 \bigr] + \sqrt{\frac{8\sigma^2}{n\delta}}}
\]
occurs with a probability bounded by $\frac{\delta}{8}$. Recall again
that for two distances $d_1$ and $d_2$ such that $d_1 \le cd_2$ one has
$ \gamma_2(\G,d_1)\le c\gamma_2(\G,d_2) $. Then Theorem \ref
{teoalternative} applies with
\[
\Gamma_{\delta}=2\sqrt{2}\sqrt{\Delta^2+\mathbb{E}
\bigl[Y^2 \bigr] + \sqrt{\frac
{8\sigma^2}{n\delta}}}\gamma_2(
\G,d_{\infty})
\]
and $\Gamma_{\delta} \ge\Delta\ge\operatorname{diam}_d(\F)$.

\begin{cor}
Consider the setup described above. Let $n\in\N$, $\delta\in
(0,1/3)$ and $\alpha=\sqrt{2\log(\delta^{-1})/(nv)}$. There exists
an integer $N_0$ and a universal
constant $C$ such that, for all $n \ge N_0$, with probability
at least $1-3\delta$,
\begin{eqnarray*}
&& R(\wh{g}) - R\bigl(g^*\bigr)
\\
&&\qquad  \le 12\sqrt{\frac{2v\log(\delta^{-1})}{n}}
\\
&&\quad\qquad{}+C\sqrt{\log\biggl(\frac
{8}{\delta}
\biggr)}\sqrt{\frac{\Delta^2+\mathbb{E}
[Y^2 ]+8\sigma
^2/(n \delta)}{n}}\int_0^\Delta
\sqrt{\log N_{\infty}(\G,\epsilon)}\,d\epsilon.
\end{eqnarray*}
\end{cor}

\begin{pf}
Apply Theorem \ref{teoalternative} and note that the condition holds
for sufficiently large $n$.
\end{pf}

The bound of the corollary essentially matches the best rates of
convergence one can get even in the case of bounded regression
under such general conditions. For special cases, such as linear
regression, better bounds may be proven for other methods; see
Audibert and Catoni \cite{AuCa11},
Hsu and Sabato
\cite{HsuSa13} and
Minsker \cite{Min13}.

\subsection{k-means clustering under heavy-tailed distribution}
\label{seckmeans}

In \textit{$k$-means\break clustering}---or \textit{vector quantization}---one
wishes to represent a distribution by a finite number of points.
Formally, let $X$ be a random vector taking values in $\R^m$
and let $P$ denote the distribution of $X$.
Let $k\ge2$ be a positive integer that we fix for the rest of
the section. A clustering scheme is given by
a set of $k$ cluster centers $C=\{y_1,\ldots,y_k\} \subset\R^m$
and a \textit{quantizer} $q:\R^m \to C$. Given a
\textit{distortion measure} $\ell: \R^m\times\R^m \to[0,\infty)$,
one wishes to find $C$ and $q$ such that the expected
distortion
\[
D_k(P,q) = \EXP\ell\bigl(X,q(X)\bigr)
\]
is as small as possible. The minimization problem is
meaningful whenever $\EXP\ell(X,0) <\infty$ which we assume
throughout. Typical distortion measures are of the form
$\ell(x,y)=\llVert x-y\rrVert ^{\alpha}$ where $\llVert
\cdot\rrVert $ is a norm on $\R^m$
and $\alpha>0$ (typically $\alpha$ equals $1$ or $2$).
Here, for concreteness and simplicity, we assume
that $\ell$ is the Euclidean distance $\ell(x,y)=\llVert
x-y\rrVert $
though the results may be generalized in a straightforward manner
to other norms.
In a way equivalent to the arguments of Section~\ref{secl2},
the results may be
generalized to the
case of the quadratic distortion $\ell(x,y)=\llVert x-y\rrVert ^2$.
In order to avoid
repetition of arguments, the details are omitted.

It is not difficult to see that if $\EXP\llVert X\rrVert <
\infty$, then
there exists a (not necessarily unique) quantizer
$q^*$ that is optimal, that is, $q^*$ is such that for all clustering
schemes $q$,
\[
D_k(P,q) \ge D_k\bigl(P,q^*\bigr) \stackrel{\mathrm{def}}{=}D_k^*(P).
\]
It is also clear that $q^*$ is a \textit{nearest neighbor quantizer},
that is,
\[
\bigl\llVert{x-q^*(x)}\bigr\rrVert=\min_{y_i \in C}\llVert
{x-y_i}\rrVert.
\]
Thus, nearest neighbor quantizers are determined by their cluster
centers $C=\{y_1,\ldots,y_k\}$. In fact, for all quantizers
with a particular set $C$ of cluster centers, the corresponding
nearest neighbor quantizer has minimal distortion and, therefore,
it suffices to restrict our attention to nearest neighbor quantizers.

In the problem of empirical quantizer design, one is given
an i.i.d. sample $X_1,\ldots,X_n$ drawn from the distribution $P$
and one's aim is to find a quantizer $q_n$ whose distortion
\[
D_k(P,q_n) = \EXP\bigl[ \bigl\llVert
X-q_n(X)\bigr\rrVert\mid X_1,\ldots,X_n
\bigr]
\]
is as close to $D_k^*(P)$ as possible. A natural strategy is to
choose a quantizer---or equivalently, a set $C$ of cluster centers---by
minimizing the \textit{empirical distortion}
\[
D_k(P_n,q) = \frac{1}{n} \sum
_{i=1}^n \bigl\llVert X_i-q(X_i)
\bigr\rrVert= \frac{1}{n} \sum_{i=1}^n
\min_{j=1,\ldots,k}\llVert X_i-y_j\rrVert,
\]
where $P_n$ denotes the standard empirical distribution
based on $X_1,\ldots,X_n$. If $\EXP\llVert X\rrVert <\infty
$, then
the empirically optimal quantizer
asymptotically minimizes the distortion. More precisely,
if $q_n$ denotes the empirically optimal
quantizer [i.e., $q_n=\argmin_q D_k(P_n,q)$], then
\[
\lim_{n\to\infty} D_k(P,q_n) =
D_k^*(P) \qquad\mbox{with probability $1$};
\]
see Pollard \cite{Pol81,Pol82a}
and Abaya and Wise \cite{AbWi84} (see also
Linder \cite{Lin02}). The rate of convergence of
$D_k(P,q_n)$ to $D_k^*(P)$ has\vspace*{1pt} drawn considerable attention;
see, for example, Pollard \cite{Pol82b},
Bartlett, Linder and Lugosi \cite{BaLiLu98},
Antos \cite{Ant05},
Antos, Gy\"orfi and Gy\"orgy \cite{AnGyGy05},
Biau, Devroye and Lugosi \cite{BiDeLu08},
Maurer and Pontil \cite{MaPo10} and
Levrard \cite{Lev13}.
Such rates are typically studied under the assumption that
$X$ is almost surely bounded. Under such assumptions, one can show
that
\[
\EXP D_k(P,q_n) - D_k^*(P) \le C(P,k,m)
n^{-1/2},
\]
where the constant $C(P,k,m)$ depends on $\esssup\llVert X\rrVert $,
$k$, and the dimension
$m$. The value of the constant has mostly been investigated in the
case of quadratic loss $\ell(x,y)=\llVert x-y\rrVert ^2$
but most proofs may be modified for the case studied here.
For the quadratic loss, one may take $C(P,k,m)$ as a constant multiple
of $B^2\min(\sqrt{k^{1-2/m}m},k)$ where $B=\esssup\llVert
X\rrVert $.

However, little is known about the finite-sample performance of empirically
designed quantizers under possibly heavy-tailed distributions.
In fact, there is no hope to extend the results cited above for
distributions with finite second moment simply because empirical
averages are poor estimators of means under such general conditions.

In the recent paper of Telgarsky and Dasgupta \cite{TeDa13}, bounds on
the excess risk under conditions on higher moments have been developed.
They prove a bound of $ \mathcal{O}(n^{-1/2+2/p}) $ for the excess
distortion where $p$ is the number of moments of $\llVert X\rrVert $
that are
assumed to be finite.
Here, we show that there exists an empirical quantizer $\widehat{q}_n$
whose excess distortion
$D_k(P,\wh{q}_n) - D_k^*(P)$ is of the order of $n^{-1/2}$ (with high
probability)
under the only assumption that $\mathbb{E} [\llVert {X}\rrVert ^2 ]$
is finite.
This may be achieved by choosing a quantizer that minimizes
Catoni's estimate of the distortion.

The proposed empirical quantizer uses two parameters that depend on
the (unknown) distribution of $X$. For simplicity, we assume that
upper bounds for these two parameters are available. (Otherwise either
one may try to estimate them or, as the sample size grows, use
increasing values for these parameters. The details go beyond
the scope of this paper.)

One of these parameters is the second moment $\operatorname{Var}
(X )=\mathbb{E} [\llVert {X-\mathbb{E} [X
]}\rrVert ^2 ]$ and let $V$ be
an upper bound. The other parameter $\rho>0$ is an upper bound
for the norm of the possible cluster centers.
The next lemma offers an estimate.

\begin{lem}[(Linder \cite{Lin02})]\label{lem-local}
Let $2\le j\le k$ be the unique integer
such that $D^*_k=\cdots=D^*_j<D^*_{j-1}$
and define $\varepsilon= (D^*_{j-1}-D^*_j)/2$.
Let $(y_1,\dots,y_j)$ be a set of cluster centers such that
the distortion of the corresponding quantizer is less than
$D^*_j+\varepsilon$.
Let $B_r =\{x:\llVert {x}\rrVert \le r\}$ denote the closed
ball of radius $r>0$
centered at the origin.
If $\rho>0$, is such that:
\begin{itemize}
\item$\frac{\rho}{10} P(B_{\sfrac{\rho}{10}}) > 2 \EXP\llVert
{X}\rrVert $,
\item$P(B_{2\rho/5})> 1- \frac{\varepsilon^2}{4\mathbb{E}
[\llVert {X}\rrVert ^2 ]}$,
\end{itemize}
then for all $1\le i \le k$, $\llVert {y_i}\rrVert \le\rho$.
\end{lem}

Now we are prepared to describe the proposed empirical quantizer.
Let $\C_{\rho}$ be the set of all collections
$C=\{y_1,\ldots,y_k\} \in(\R^m)^k$ of cluster centers with
$\llVert y_j\rrVert \le\rho$ for all $j=1,\ldots,k$.
For each $C\in\C_{\rho}$, denote by $q_C$ the corresponding
quantizer. Now for all $C\in\C_{\rho}$, we may calculate
Catoni's mean estimator of the distortion
$D(P,q_C)= \EXP\llVert X-q_C(X)\rrVert = \EXP\min_{j=1,\ldots,k}\llVert X_i-y_j\rrVert $
defined as
the unique value $\mu\in\R$ for which
\[
\frac{1}{n\alpha} \sum_{i=1}^n \phi
\Bigl(\alpha\Bigl(\min_{j=1,\ldots,k} \llVert X_i-y_j
\rrVert-\mu\Bigr) \Bigr) =0,
\]
where we use the parameter value $\alpha=\sqrt{2/(nkV)}$.
Denote this estimator by $\wh{D}(P_n,q_C)$ and let $\wh{q}_n$
be any quantizer minimizing the estimated distortion.
An easy compactness argument shows that such a minimizer exists.

The main result of this section is the following bound
for the distortion of the chosen quantizer.

\begin{theorem}
Assume that $\operatorname{Var} (X ) \le V <\infty$ and
$n\ge m$. Then, with
probability at least $1-\delta$,
\[
D_k(P,\wh{q}_n)-D_k\bigl(P,q^*\bigr) \le C
\biggl(\log\frac{1}{\delta} \biggr) \biggl(\sqrt{\frac{Vk}{n}}+\sqrt{
\frac{mk}{n}} \biggr)+O \biggl(\frac{1}{n} \biggr),
\]
where the constant $C$ only depends on $\rho$.
\end{theorem}

\begin{pf}
The result follows from Theorem \ref{teomain}. All we need to check
is that $\operatorname{Var} (\min_{j=1,\ldots,k} \llVert
X-y_j\rrVert )$ is bounded by $kV$
and estimate the covering numbers of the class of functions
\[
\F_{\rho} = \Bigl\{ f_C(x)=\min_{y \in C}
\llVert x-y\rrVert: C\in\C_{\rho} \Bigr\}.
\]
The variance bound follows simply by the fact that for all $C\in\C$,
\begin{eqnarray*}
\operatorname{Var} \Bigl(\min_{j=1,\ldots,k} \llVert X-y_j
\rrVert\Bigr) &\le& \sum_{i=1}^k
\operatorname{Var} \bigl(\llVert{X-y_i}\rrVert\bigr)
\\
&\le& \sum_{i=1}^k \mathbb{E} \bigl[
\llVert{X-\EXP X}\rrVert^2 \bigr]+ \llVert{\EXP X-y_i}
\rrVert^2-\mathbb{E} \bigl[\llVert{X-y_i}\rrVert
\bigr]^2
\\
&\le& kV.
\end{eqnarray*}
In order to use the bound of Theorem \ref{teomain}, we need to
bound the covering numbers of the class $\F_{\rho}$ under both
metrics $d$ and $D$. We begin with the metric
\[
D(f_C,f_{C'})=\sup_{x \in\R^m}\bigl\llvert
f_C(x)-f_{C'}(x) \bigr\rrvert.
\]
The notation $B_z(\epsilon,d)$ refers to the ball under the metric $d$
of radius $\epsilon$ centered at~$z$.
Let $Z$ be a subset of $B_{\rho}$ such that
\[
\mathcal{B}_{B_{\rho}}:=\bigl\{B_{z}\bigl(\epsilon,\llVert\cdot
\rrVert\bigr): z \in Z\bigr\}
\]
is a covering of the set $B_{\rho}$ by balls of radius $\epsilon$
under the Euclidean norm. Let $C \in\mathcal{C}_{\rho}$ and
associate to any $y_i \in C$ one of the centers in $Z$ such that $
\llVert {y_i-z_i}\rrVert \le\epsilon$. If there is more
than one possible
choice for $z_i $, we pick one of them arbitrarily. We denote by $
q_{C'} $ the nearest neighbor quantizer with codebook $C'=(z_i)_i $.
Finally,\vspace*{2pt} let $S_i = q_{C'}^{-1}(z_i)$.
Now clearly, $\forall i, \forall x \in S_i $
\begin{eqnarray*}
f_C(x)-f_{C'}(x) &=& \min_{1\le j \le k} \llVert
x-y_j\rrVert-\min_{1\le j \le
k} \llVert
x-z_j\rrVert
\\
& = & \min_{1\le j \le k} \llVert x-y_j\rrVert- \llVert
x-z_i\rrVert
\\
&\le& \llVert x-y_i\rrVert- \llVert x-z_i\rrVert\le
\epsilon
\end{eqnarray*}
and similarly, $f_{C'}(x)-f_{C}(x)\le\epsilon$. Then $f_C \in
B_{f_{C'}}(\epsilon,D)$ and
\[
\mathcal{B}_{\F_{\rho}}:= \bigl\{B_{f_{C}}(\epsilon,D): C \in
Z^{k} \bigr\}
\]
is a covering of $ \F_{\rho} $. Since $Z$ can be taken such that $
\llvert Z\rrvert = N_{\llVert \cdot\rrVert
}(B_{\rho},\epsilon)$ we obtain
\[
N_{d}(\F_{\rho},\epsilon) \le N_{D}(
\F_{\rho},\epsilon) \le N_{\llVert
\cdot\rrVert }(B_{\rho},
\epsilon)^k.
\]
By standard estimates on the covering numbers of the ball $B_{\rho}$
by balls of size $\epsilon$ under the Euclidean metric,
\[
N_{\llVert \cdot\rrVert }(B_\rho,\epsilon) \le\biggl(\frac
{4\rho}{\epsilon
}
\biggr)^m
\]
(see, e.g., Matousek \cite{mat02}).
In other words, there exists a universal constant $L$ and constants
$C_\rho$ and $C'_\rho$ that depends only on $\rho$ such that
\begin{eqnarray*}
\gamma_2(\F_{\rho},d) &\le& L\int_{0}^{2\rho}
\sqrt{\log N_{d}(\F_\rho,\epsilon)} \,d\epsilon\le
C_\rho\sqrt{km},
\end{eqnarray*}
and
\begin{eqnarray*}
\gamma_1(\F_{\rho},D) &\le&  L\int
_{0}^{2\rho} \log N_{D}(
\F_\rho,\epsilon) \,d\epsilon\le C'_\rho km.
\end{eqnarray*}

Theorem \ref{teomain} may now be applied to the class $\F_{\rho}$.
\end{pf}

\section{Simulation study}
\label{secsim}

In this closing section, we present the results of two simulation
exercises that assess the performance of the
estimators developed in this work.

\subsection{$L_2$ regression}

The first application is an $L_2$ regression exercise.
Data are simulated from a linear model with heavy-tailed errors and
the $L_2$ regression procedure based on Catoni's risk minimizer
introduced in Section~\ref{secl2}
is used for estimation. The procedure is benchmarked against regular
(``vanilla'') $L_2$ regression based on the minimization of the
empirical $L_2$ loss.

The simulation exercise is designed as follows.
We simulate $(Z_1,Y_1)$,\break $(Z_2,Y_2), \ldots, (Z_n,Y_n)$ i.i.d. pairs of
random variables in $\mathbb R^5 \times\mathbb R$.
The vector $Z_i$ of explanatory variables is drawn from a multivariate
normal distribution with zero mean, unit variance
and correlation matrix equal to an equi-correlation matrix with
correlation $\rho=0.9$.
The response variable $Y_i$ is generated as
\[
Y_i = Z_i^T\theta^* +
\epsilon_i ,
\]
where the parameter vector $\theta^*$ is set to $(0.25, -0.25, 0.50,
0.70, -0.75)$
and $\epsilon_i$ is a zero mean error term.
The error term $\epsilon_i$ is drawn from a Pareto distribution with
tail parameter $\beta$ and is appropriately recentered in order to
have zero mean.
As it is well known, the tail parameter $\beta$ determines which
moments of the Pareto random variable are finite.
More specifically, the moment of order $k$ exists only if $k < \beta$.
The focus is on finding the value of $\theta$ which minimizes the
$L_2$ risk
\[
\mathbb E \bigl\llvert Y - Z_i^T\theta\bigr\rrvert
^2 .
\]
The parameter $\theta$ is estimated using the Catoni and the vanilla
$L_2$ regressions.
Let $\widehat{R}_C(\theta)$ denote the solution of the equation
\[
\widehat r_\theta(\mu)= {1 \over n\alpha} \sum
_{i=1}^n \phi\bigl( \alpha\bigl( \bigl\llvert
Y_i - Z_i^T \theta\bigr\rrvert
^2 - \mu\bigr) \bigr) = 0 .
\]
Then the Catoni $L_2$ regression estimator is defined as
\[
\wh{\theta}_{n C} = \arg\min_{\theta}
\widehat{R}_C(\theta) .
\]
The vanilla $L_2$ regression estimator is defined as the minimizer of
the empirical $L_2$ loss,
\[
\wh{\theta}_{n V} = \arg\min_{\theta}
\widehat{R}_V(\theta) = \arg\min_{\theta}
{1 \over n}\sum_{i=1}^n \bigl
\llvert Y_i - Z_i^T\theta\bigr\rrvert
^2 ,
\]
which is the classical least squares estimator.
The precision of each estimator is measured by their excess risk
\begin{eqnarray*}
R(\wh{\theta}_{n C}) - R\bigl(\theta^*\bigr) & = & \mathbb E \bigl
\llvert Y - Z^{T}\wh{\theta}_{n C}\bigr\rrvert
^2 - \mathbb E \bigl\llvert Y - Z^{T}\theta^*\bigr\rrvert
^2,
\\
R(\wh{\theta}_{n V}) - R\bigl(\theta^*\bigr) & = & \mathbb E \bigl
\llvert Y - Z^{T}\wh{\theta}_{n V}\bigr\rrvert
^2 - \mathbb E \bigl\llvert Y - Z^{T}\theta^*\bigr\rrvert
^2 .
\end{eqnarray*}
We estimate excess risk by simulation. For each replication of the
simulation exercise, we estimate the risk of the estimators and the
optimal risk
using sample averages based on an i.i.d. sample $(Z_1',Y_1'),\ldots,
(Z_m',Y_m')$ that is independent of the one used for estimation, that is,
%
\begin{eqnarray}
\label{eqnladcen} \widetilde{R}(\wh{\theta}_{n C}) & = &
{1 \over m} \sum_{i=1}^m \bigl
\llvert Y_i' - {Z_i'}^{T}
\wh{\theta}_{n C}\bigr\rrvert^2,\nonumber
\\
\widetilde{R}(\wh{\theta}_{n V}) & = & {1 \over m} \sum
_{i=1}^m \bigl\llvert Y_i'
- {Z_i'}^{T}\wh{\theta}_{n V}\bigr
\rrvert^2,
\\
\widetilde{R}\bigl(\theta^*\bigr) & = & {1 \over m} \sum
_{i=1}^m \bigl\llvert Y_i'
- {Z_i'}^{T} \theta^*\bigr\rrvert
^2 .\nonumber
\end{eqnarray}
The simulation experiment is replicated for different values of the
Pareto tail parameter $\beta$ ranging from $2.01$ to $6.01$ and
different values of the sample size $n$, ranging from $50$ to 1000.
For each combination of the tail parameter $\beta$ and sample size
$n$, the experiment is replicated 1000 times.

%
\begin{figure}[t]

\includegraphics{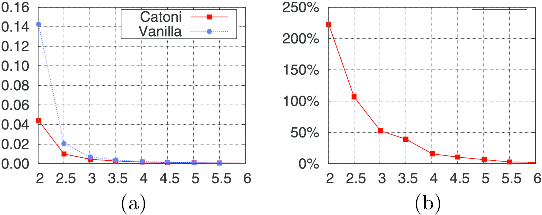}

\caption{$L_2$ regression parameter estimation.
The figure plots the excess risk of the Catoni and vanilla
$L_2$ regression parameter estimators \textup{(a)} and the percentage
improvement of the Catoni procedure relative to the vanilla \textup{(b)} as a
function of the tail parameter $\beta$ for a sample size $n$ equal to~500.}\label{figreg}
\end{figure}

\begin{table}[b]
\tabcolsep=0pt
\caption{Relative performance of the Catoni $L_2$ parameter estimator}\label{tblreg}
\begin{tabular*}{\tablewidth}{@{\extracolsep{\fill}}@{}ld{4.2}d{3.2}d{3.2}d{3.2}d{3.2}d{3.2}@{}}
\hline
$\bolds{\beta}$ & \multicolumn{1}{c}{$\bolds{n=50}$} & \multicolumn{1}{c}{$\bolds{n=100}$} &
\multicolumn{1}{c}{$\bolds{n=250}$} & \multicolumn{1}{c}{$\bolds{n=500}$} &
\multicolumn{1}{c}{$\bolds{n=750}$} & \multicolumn{1}{c}{$\bolds{n=1000}$}
\\
\hline
2.01&3872.10&440.50&171.30&222.70&218.20&142.80\\
2.50&169.20&158.70&151.50&106.70&91.70&57.40\\
3.01&137.60&178.00&89.00&52.50&62.70&63.50\\
3.50&54.40&20.90&41.30&39.20&38.10&33.50\\
4.01&30.20&44.40&25.50&15.70&16.30&15.90\\
4.50&16.50&12.10&11.30&10.60&6.90&13.70\\
5.01&10.20&7.80&10.20&6.40&5.70&3.10\\
5.50&6.00&14.80&3.90&2.90&2.10&2.20\\
6.01&3.90&1.90&2.70&2.10&1.90&1.40\\
\hline
\end{tabular*}
\tabnotetext[]{}{The table reports the percentage improvement of the excess
risk of the Catoni $L_2$ regression estimator relative to the
vanilla procedure for different values of the tail parameter $\beta$
and sample size~$n$.}
\end{table}

Figure~\ref{figreg} displays the Monte Carlo estimate of the excess
risk of the Catoni and benchmark regression estimators as functions of
the tail parameter $\beta$ when the sample size $n$ is equal to $500$.
The left panel shows the level of the excess risks
$R(\wh{\theta}_{n C}) - R(\theta^*)$ and $R(\wh{\theta}_{n V}) -
R(\theta^*)$
as a function of $\beta$
and the right one shows the percentage improvement of the excess risk
of the Catoni procedure over the benchmark calculated as
$ (R(\wh{\theta}_{n V}) - R(\wh{\theta}_{n C}) )/
(R(\wh{\theta}_{n C}) - R(\theta^*) )$.
When the tails are not excessively heavy (high values of $\beta$) the
difference between the procedures is small.
As the tails become heavier (small values of $\beta$), the risks of
both procedures increase.
Importantly, the Catoni estimator becomes progressively more efficient
as the tails become heavier and becomes significantly more efficient
when the tail parameter is close to $2$.
%
Detailed results for different values of $n$ are reported in Table~\ref
{tblreg}.
Overall, the Catoni $L_2$ regression estimator never performs worse
than the benchmark, and it is substantially better
when the tails of the data are heavy.

\subsection{$k$-means}

In the second experiment, we carry out a $k$-means clustering exercise.
Data are simulated from a heavy-tailed mixture distribution and then
cluster centers are chosen
by minimizing Catoni's estimate of the $L_2$ distortion.
The performance of the algorithm is benchmarked against the
(``vanilla'') $k$-means algorithm procedure
where the distortion is estimated by the standard empirical average.

The simulation exercise is designed as follows.
An i.i.d. sample of random vectors $X_1, \ldots, X_n$ in $\mathbb
R^2$ is
drawn from a four-component mixture distribution with equal weights.
The means of the mixture components are $(5,5), (-5,5), (-5,-5)$ and $(5,-5)$.
Each component of the mixture is made up of two appropriately centered
independent draws from a Pareto distribution with tail parameter $\beta$.
The cluster centers obtained by the $k$-means algorithm based on
Catoni and the vanilla $k$-means algorithm are denoted, respectively, by
$\wh{q}_{n C}$ and $\wh{q}_{n V}$.
%
(Since finding the empirically optimal cluster centers is
computationally prohibitive, we use the well-known iterative
optimization procedure ``$k$-means'' for the vanilla version and a
similar variant for the Catoni scheme.)
Analogously to the previous exercise, we summarize the performance of
the clustering procedures using
the excess risk of the algorithms, that is,
\begin{eqnarray*}
&&D_k( P , \wh{q}_{n C} ) - D_k\bigl( P , q^*
\bigr),\qquad D_k( P , \wh{q}_{n V} ) - D_k\bigl(
P , q^* \bigr),
\end{eqnarray*}
where $q^*$ denotes the means of the mixture components.
We estimate excess risk by simulation.
We compute the distortion of the quantizers using an
i.i.d. sample $X_1',\ldots, X_m'$ of vectors that is independent of
the ones used for estimation, that is,
%
\begin{eqnarray}
\label{eqnkmeanscen} D_k(P_m,\wh{q}_{n C}) & = &
{1 \over m} \sum_{i=1}^m \min
_{j=1,\ldots,k} \bigl\llVert X_i' -
\wh{q}_{n C}\bigl(X_i'\bigr) \bigr\rrVert
^2,\nonumber
\\
D_k(P_m,\wh{q}_{n V}) & = &
{1 \over m} \sum_{i=1}^m \min
_{j=1,\ldots,k} \bigl\llVert X_i' -
\wh{q}_{n V}\bigl(X_i'\bigr) \bigr\rrVert
^2,
\\
D_k\bigl(P_m,q^*\bigr) & = & {1 \over m}
\sum_{i=1}^m \min_{j=1,\ldots,k}
\bigl\llVert X_i' - q^*\bigl(X_i'
\bigr) \bigr\rrVert^2.\nonumber
\end{eqnarray}
The experiment is replicated for different values of the tail parameter
$\beta$ ranging from $2.01$ to $6.01$ and different values of the
sample size $n$ ranging from $50$ to 1000.
For each combination of tail parameter $\beta$ and sample size $n$ the
experiment is replicated 1000 times.

\begin{figure}[t]

\includegraphics{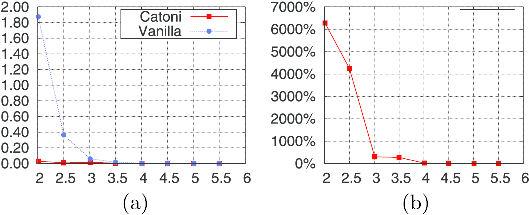}

\caption{$k$-means quantizer estimation.
The figure plots the excess risk of the Catoni and vanilla
$k$-means quantizer estimator \textup{(a)} and the percentage improvement of
the Catoni procedure relative to the vanilla \textup{(b)} as a function of the
tail parameter $\beta$ for a sample size $n$ equal to 500.}\label{figkmeans}
\end{figure}

\begin{table}[b]
\tabcolsep=0pt
\caption{Relative performance of the Catoni $k$-means quantizer estimator}\label{tblkmeans}
\begin{tabular*}{\tablewidth}{@{\extracolsep{\fill}}@{}ld{3.2}d{4.2}d{4.2}d{4.2}d{4.2}d{5.2}@{}}
\hline
$\bolds{\beta}$ & \multicolumn{1}{c}{$\bolds{n=50}$} & \multicolumn{1}{c}{$\bolds{n=100}$} &
\multicolumn{1}{c}{$\bolds{n=250}$} & \multicolumn{1}{c}{$\bolds{n=500}$} &
\multicolumn{1}{c}{$\bolds{n=750}$} & \multicolumn{1}{c}{$\bolds{n=1000}$}
\\
\hline
2.01&823.80&2180.40&3511.60&6278.90&7858.70&10{,}684.60\\
2.50&404.50&1007.40&2959.80&4255.40&6828.60&9093.60\\
3.01&301.10&312.20&286.80&298.60&813.60&1560.20\\
3.50&129.60&188.60&213.30&271.40&448.60&410.00\\
4.01&73.80&30.90&26.80&20.30&18.20&13.10\\
4.50&27.60&22.90&16.50&11.70&9.50&10.10\\
5.01&16.40&10.80&11.60&8.70&6.00&7.20\\
5.50&9.00&6.80&9.20&5.00&4.10&4.00\\
6.01&3.50&4.70&5.00&2.70&3.20&3.10\\
\hline
\end{tabular*}
\tabnotetext[]{}{The table reports the improvement of the Catoni $k$-means
quantizer estimator relative to the
vanilla procedure for different values of the tail parameter $\beta$
and sample size $n$.}
\end{table}

Figure~\ref{figkmeans} displays the Monte Carlo estimate of excess
risk of the Catoni and benchmark estimators as a function of tail
parameter $\beta$ for $n=500$.
The left panel shows the estimated excess risk while the right panel
shows the percentage improvement of the excess risk of the Catoni
procedure, calculated as
$ (D_k( P , \wh{q}_{n V} )-D_k( P , \wh{q}_{n C} ) )/
(D_k( P , \wh{q}_{n C} ) - D_k( P , q^* ) )$.

The overall results are analogous to the ones of the $L_2$ regression
application.
When the tails of the mixture are not excessively heavy (high values of
$\beta$) the difference in the procedures is small.
As the tails become heavier (small values of $\beta$), the risk of both
procedures increases,
but the Catoni algorithm becomes progressively more efficient.
The percentage gains for the Catoni procedure are substantial when the
tail parameter is smaller than $4$.
Table~\ref{tblkmeans} reports detailed results for different values
of $n$.
As in the $L_2$ regression simulation study, the Catoni $k$-means
algorithm never performs worse than the benchmark and it is
substantially better when the tails of the mixture are heavy.

\begin{appendix}\label{app}
\section*{Appendix}
\subsection{A chaining theorem}
The following result is a version of standard bounds based on
``generic chaining''; see Talagrand \cite{talagrand}. We include
the proof for completeness.

Recall that if $\psi$ is a nonnegative increasing convex function
defined on $\R_+$ with $\psi(0)=0$, then
the Orlicz norm of a random variable $X$ is defined by
\[
\llVert X\rrVert_{\psi} = \inf\biggl\{c>0 : \mathbb{E} \biggl[\psi
\biggl(\frac{\llvert X\rrvert }{c} \biggr) \biggr]\le1 \biggr\}.
\]
We consider Orlicz norms defined by
\[
\psi_1(x) =\exp(x)-1 \quad\mbox{and} \quad\psi_2(x)=
\exp\bigl(x^2\bigr)-1.
\]
For further information on Orlicz norms, see \cite{vaWe96},
Chapter~2.2. First, $\llVert X\rrVert _{\psi
_{1}}\le\llVert X\rrVert _{\psi_{2}}\sqrt{\log(2)}$ holds.
Also note that, by Markov's inequality, $\llVert X\rrVert
_{\psi_{1}}\le c$ implies
that $\PROB\{\llvert X\rrvert >t\}\le2e^{-t/c}$ and
similarly, if
$\llVert X\rrVert _{\psi_{2}}\le c$, then $\PROB\{\llvert X\rrvert
>t\}\le2e^{-t^2/c^2}$.
Then
\begin{eqnarray}
\label{eq-orlic} X &\le&\llVert X\rrVert_{\psi_{1}}\log\bigl(2
\delta^{-1}\bigr)\qquad\mbox{with probability at least } 1-\delta,
\nonumber\\[-8pt]\\[-8pt]\nonumber
X &\le&\llVert X\rrVert_{\psi_{2}}\sqrt{\log\bigl(2\delta^{-1}
\bigr)}\qquad\mbox{with probability at least } 1-\delta.
\nonumber
\end{eqnarray}
Recall the following definition (see, e.g., \cite{talagrand}, Definition~1.2.3).
Let $T$ be a (pseudo) metric space. An increasing sequence
$(\mathcal{A}_n)$ of partitions of $T$ is called \textit{admissible} if
for all $n=0,1,2,\ldots, \#\mathcal{A}_n \le2^{2^n}$.
For any $t\in T$, denote by $A_n(t)$ the unique element
of $\mathcal{A}_n$ that contains $t$.
Let $\Delta(A)$ denote the diameter of the set $A\subset T$.
Define, for $\beta=1,2$,
\[
\gamma_{\beta}(T,d) =\inf_{\mathcal{A}_n}\sup
_{t \in T} \sum_{n \ge0} 2^{n/\beta
}
\Delta\bigl(A_n(t)\bigr),
\]
where the infimum is taken over all admissible sequences. First of all,
we know from \cite{talagrand}, equation (1.18), that there exists a
universal constant $L$ such that
%
\begin{equation}
\label{eqdudley} \gamma_{\beta}(T,d) \le L \int_0^{\mathrm{diam}_d(T)}
\bigl(\log N_d(T,\varepsilon) \bigr)^{\sfrac{1}{\beta}}\,d\varepsilon.
\end{equation}
%

\begin{theorem}
\label{teochaining}
Let $(X_t)_{t\in T}$ be a stochastic process indexed by a set $T$ on which
two (pseudo) metrics, $d_1$ and $d_2$, are defined such that $T$ is bounded
with respect to both metrics.
Assume that for any $s,t \in T$ and for all $x>0$,
\[
\PROB\bigl\{\llvert X_s-X_t\rrvert>x \bigr\} \le2\exp
\biggl(-\frac{1}{2}\frac{x^2}{d_2(s,t)^2+d_1(s,t)x} \biggr).
\]

Then for all $t\in T$,
\[
\Bigl\llVert\sup_{s \in T} \llvert X_s-X_t
\rrvert\Bigr\rrVert_{\psi_{1}} \le L \bigl( \gamma_1(T,d_1)
+ \gamma_2(T,d_2) \bigr)
\]
with $L\le384 \log(2)$.
\end{theorem}

The proof of Theorem \ref{teochaining} uses the following lemma.

\begin{lem}[(\cite{vaWe96}, Lemma 2.2.10)]\label{lemvvw}
Let $a,b>0$ and assume that the random variables $X_1,\dots,X_m$ satisfy,
for all $x>0$,
\[
\PROB\bigl\{\llvert X_i\rrvert>x\bigr\} \le2\exp\biggl(-
\frac
{1}{2}\frac{x^2}{b+ax} \biggr).
\]
Then
\[
\Bigl\llVert\max_{1\le i \le m} X_i\Bigr\rrVert
_{\psi_{1}} \le48 \bigl(a\log(1+m) +\sqrt{b} \sqrt{\log(1+m)} \bigr).
\]
\end{lem}

\begin{pf*}{Proof of Theorem \ref{teochaining}}
Consider an admissible sequence $(\mathcal{B}_n)_{n \ge0}$ such that
for all $t \in T$,
\[
\sum_{n \ge0} 2^n \Delta_1
\bigl(B_n(t)\bigr) \le2 \gamma_1(T,d_1)
\]
and an admissible sequence $(\mathcal{C}_n)_{n \ge0}$ such that
for all $t \in T$,
\[
\sum_{n \ge0} 2^{n/2} \Delta_2
\bigl(C_n(t)\bigr) \le2 \gamma_2(T,d_2).
\]
Now we may define an admissible sequence by intersection of the
elements of $(\mathcal{B}_{n-1})_{n \ge1}$ and $(\mathcal
{C}_{n-1})_{n \ge1}$:
set $\mathcal{A}_0=\{T\}$ and let
\[
\mathcal{A}_n =\{ B\cap C : B \in\mathcal{B}_{n-1} \mbox{ and } C\in
\mathcal{C}_{n-1}\}.
\]
$(\mathcal{A}_n)_{n \ge0}$ is an admissible sequence because
each $\mathcal{A}_n$ is increasing and contains at most
$(2^{2^{n-1}})^2=2^{2^n}$ sets.
Define a sequence of finite sets $T_0=\{t\}\subset T_1\subset\cdots
\subset T$
such that $T_n$ contains a single point in each set of $\mathcal{A}_n$.
For any $s \in T$, denote by $\pi_n(s)$ the unique element of
$T_n$ in $A_n(s)$.
Now for any $s \in T_{k+1}$, we write
\[
X_s-X_t= \sum_{k=0}^\infty
(X_{\pi_{k+1}(s)}-X_{\pi_{k}(s)} ).
\]
Then, using the fact that $\llVert \cdot\rrVert _{\psi_{1}}$
is a norm and
Lemma \ref{lemvvw},
\begin{eqnarray*}
&& \Bigl\llVert\sup_{s \in T} \llvert
X_s-X_t\rrvert\Bigr\rrVert_{\psi_{1}}
\\
&&\qquad  \le \sum_{k=0}^\infty\Bigl\llVert\max
_{s \in T_{k+1}} \llvert X_{\pi_{k+1}(s)}-X_{\pi_{k}(s)}\rrvert\Bigr
\rrVert_{\psi_{1}}
\\
&&\qquad \le 48\sum_{k=0}^\infty
\bigl(d_1\bigl(\pi_{k+1}(s),\pi_{k}(s)\bigr)\log
\bigl(1+2^{2^{k+1}}\bigr)
\\
&&\quad\qquad{} +d_2\bigl(\pi_{k+1}(s),
\pi_{k}(s)\bigr)\sqrt{\log\bigl(1+2^{2^{k+1}}\bigr)} \bigr).
\end{eqnarray*}
Since $(\mathcal{A}_n)_{n \ge0}$ is an increasing sequence,
$\pi_{k+1}(s)$ and $\pi_{k}(s)$ are both in $A_k(s)$. By construction,
$A_k(s) \subset B_k(s)$ and, therefore, $d_1(\pi_{k+1}(s),\pi_{k}(s))
\le\Delta_1(B_k(s))$. Similarly, we have $d_2(\pi_{k+1}(s),\pi
_{k}(s)) \le\Delta_2(C_k(s))$.
Using $\log(1+2^{2^{k+1}}) \le4\log(2) 2^k$, we get
\begin{eqnarray*}
\Bigl\llVert\max_{s \in T} \llvert X_s-X_t
\rrvert\Bigr\rrVert_{\psi_{1}} &\le& 192 \log(2) \Biggl[ \sum
_{k=0}^\infty2^k \Delta_1
\bigl(B_k(s)\bigr)+\sum_{k=0}^\infty2^{k/2}
\Delta_2\bigl(C_k(s)\bigr) \Biggr]
\\
&\le& 384 \log(2) \bigl[ \gamma_{1}(T,d_1)+
\gamma_{2}(T,d_2) \bigr].
\end{eqnarray*}\upqed
\end{pf*}

\begin{proposition}
\label{cororlic}
Assume that for any $s,t \in T$ and for all $x>0$,
\[
\PROB\bigl\{\llvert X_s-X_t\rrvert>x \bigr\} \le2\exp
\biggl(-\frac{x^2}{2d_2(s,t)^2} \biggr).
\]
Then for all $t\in T$,
\[
\Bigl\llVert\sup_{s \in T} \llvert X_s-X_t
\rrvert\Bigr\rrVert_{\psi_{2}} \le L \gamma_2(T,d_2),
\]
where $L$ is a universal constant.
\end{proposition}

The proof of Proposition \ref{cororlic} is
similar to the proof of Theorem \ref{teochaining}.
One merely needs to replace Lemma \ref{lemvvw} by Lemma 2.2.2 in
\cite{vaWe96}
and proceed identically. The details are omitted.

We may use Proposition \ref{cororlic} to bound the moment generating
function of $\sup_{s \in T} \llvert X_s-X_t\rrvert $ as follows.
Set $S = \sup_{s \in T} \llvert X_s-X_t\rrvert $. Then using
$ab \le(a^2+b^2)/2$,
we have, for every $\lambda>0$,
\[
\exp(\lambda S) \le\exp\bigl(S^2/\llVert{S}\rrVert
^2_{\psi_2}+\lambda^2\llVert{S}\rrVert
^2_{\psi_2}/4 \bigr),
\]
and, therefore,
%
\begin{equation}
\label{eqcharac} \mathbb{E} \Bigl[\exp\Bigl(\lambda\sup_{s \in T}
\llvert X_s-X_t\rrvert\Bigr) \Bigr] \le2\exp\bigl(
\lambda^2 L^2 \gamma_2(T,d_2)^2/4
\bigr).
\end{equation}
\end{appendix}


\section*{Acknowledgments}
We thank the referees for their thorough reading and insightful comments that helped greatly improve
the manuscript.


%

\printaddresses
\end{document}